\pdfoutput=1 
\documentclass[11pt]{article}

\newif\ifanonymous
\newif\ifcomments
\commentsfalse
\anonymousfalse

\usepackage[letterpaper,margin=1in]{geometry}
\providecommand{\keywords}[1]{\textbf{{Keywords:}} #1.} 

\usepackage[backend=biber,style=alphabetic,sorting=ynt,maxbibnames=10,maxalphanames=2,minalphanames=2]{biblatex}
\usepackage[T1]{fontenc}    
\usepackage[english]{babel} 
\usepackage{mathtools}      
\usepackage{amsfonts}      
\usepackage{csquotes}       
\usepackage{enumitem}       
\usepackage{graphicx}       
\graphicspath{{./images/}}
\usepackage{tikz}
\usetikzlibrary{shapes.geometric,positioning}
\usepackage{booktabs}       
\usepackage{array,multirow} 
\usepackage{tabularx}       
\usepackage{amssymb}
\usepackage{diagbox}        
\usepackage{pifont}
\newcommand{\cmark}{\checkmark}%
\newcommand{\xmark}{\times}%
\usepackage{amsthm}         
\usepackage{thmtools,thm-restate}   
\usepackage{url}
\PassOptionsToPackage{hyphens}{url}

\usepackage{breakurl}
\usepackage{hyperref}  
\PassOptionsToPackage{breaklinks}{hyperref}
\usepackage{xcolor}         
\hypersetup{                
    colorlinks,
    linkcolor={red!50!black},
    citecolor={blue!50!black},
    urlcolor={blue!80!black}
}
\usepackage[capitalise]{cleveref}       

\ifcomments
    \usepackage[colorinlistoftodos,textsize=tiny]{todonotes}
    \newcommand{\aviv}[1]{\todo[inline,color=blue!40]{Aviv: #1}}
    \newcommand{\rainer}[1]{\todo[inline,color=blue!80]{Rainer: #1}}
    \newcommand{\svetlana}[1]{\todo[inline,color=blue!10]{Svetlana: #1}}
\else
    \newcommand{\aviv}[1]{}
    \newcommand{\rainer}[1]{}
    \newcommand{\svetlana}[1]{}
\fi

\newcommand{\paragraphNoSkip}[1]{\par\smallskip\noindent\textbf{#1}.}

\newtheorem{definition}{Definition}[section]

\newtheorem{remark}[definition]{Remark}

\newtheorem*{example*}{Example}

\crefname{definition}{Definition}{Definitions}
\crefname{theorem}{Theorem}{Theorems}
\crefname{claim}{Claim}{Claims}
\crefname{lemma}{Lemma}{Lemmas}
\crefname{corollary}{Corollary}{Corollaries}
\crefname{example}{Example}{Examples}
\crefname{remark}{Remark}{Remarks}
\crefname{code}{Code}{Code}

\newcommand{\define}{\stackrel{\mathclap{\tiny\mbox{def}}}{=}}

\usepackage[symbols,acronym,nonumberlist,nogroupskip,stylemods={mcols,longbooktabs},section=subsection,numberedsection]{glossaries-extra}
\makenoidxglossaries

\setabbreviationstyle[acronym]{long-short}  
\glssetcategoryattribute{acronym}{nohyper}{true} 
\newacronym[longplural={Markov decision processes}]{MDP}{MDP}{Markov decision process}
\newacronym{AI}{AI}{artificial intelligence}
\newacronym{AMM}{AMM}{automated market maker}
\newacronym{APY}{APY}{annual percentage yield}
\newacronym{APR}{APR}{per block interest rate (non compounding)}
\newacronym{ASIC}{ASIC}{application specific integrated circuit}
\newacronym{CDF}{CDF}{cumulative density function}
\newacronym{CPU}{CPU}{central processing unit}
\newacronym{DAA}{DAA}{difficulty-adjustment algorithm}
\newacronym{DQL}{DQL}{deep-Q-learning}
\newacronym{DeFi}{DeFi}{decentralized finance}
\newacronym{EIP}{EIP}{Ethereum improvement proposal}
\newacronym{ERC}{ERC}{Ethereum request for comments}
\newacronym{EVM}{EVM}{Ethereum virtual machine}
\newacronym{HUJI}{HUJI}{Hebrew University of Jerusalem, Israel}
\newacronym{LP}{LP}{liquidity provider}
\newacronym{LT}{LT}{liquidity taker}
\newacronym{MEV}{MEV}{miner-extractable value}
\newacronym{BEV}{BEV}{blockchain-extractable value}
\newacronym{ML}{ML}{machine learning}
\newacronym{OO}{OO}{order optimization}
\newacronym{PDF}{PDF}{probability density function}
\newacronym{PID}{PID}{proportional integral derivative}
\newacronym{PoS}{PoS}{proof-of-stake}
\newacronym{PoW}{PoW}{proof-of-work}
\newacronym{RAM}{RAM}{random-access memory}
\newacronym{RL}{RL}{reinforcement learning}
\newacronym{RPC}{RPC}{remote procedure call}
\newacronym{SSD}{SSD}{solid state drive}
\newacronym{TD}{TD}{total difficulty}
\newacronym{URL}{URL}{uniform resource locator}
\newacronym{USD}{USD}{United States Dollar}
\newacronym{WETH}{WETH}{Wrapped Ethereum}
\newacronym{WBTC}{WBTC}{Wrapped Bitcoin}
\newacronym{YAML}{YAML}{YAML Ain't Markup Language}
\newacronym{block-DAG}{block-DAG}{block directed-acyclic-graph}
\newacronym{geth}{geth}{Go Ethereum}
\newacronym{p2p}{p2p}{peer to peer}
\newacronym{FaaS}{FaaS}{front-running-as-a-service}
\newacronym{FSL}{FSL}{fixed spread liquidation}
\newacronym{DEX}{DEX}{decentralized exchange}
\newacronym{TVL}{TVL}{total value locked}
\newacronym{CPAMM}{CPAMM}{constant product automated market maker}
\newacronym{CFMM}{CFMM}{constant function market maker}
\newacronym{wlog}{w.l.o.g.}{without loss of generality}
\newacronym{wrt}{w.r.t.}{with regards to}
\newacronym{iid}{i.i.d.}{independent and identically distributed}
\newacronym{DoS}{DoS}{denial-of-service}
\newacronym{SoK}{SoK}{systemization of knowledge}
\newacronym{SVM}{SVM}{shareholder value maximization}
\newacronym{OFAC}{OFAC}{Office of Foreign Assets Control}
\newacronym{DAO}{DAO}{decentralized autonomous organization}
\newacronym{WPBE}{WPBE}{weak perfect Bayesian equilibrium}
\newacronym{SE}{SE}{sequential equilibrium}

\glsxtrnewsymbol[description={
    A vote in favor of the proposal (``For'').
}]{forVote}{
    \ensuremath{\texttt{F}}
}
\newcommand{\forVote}{{\gls[hyper=false]{forVote}}}

\glsxtrnewsymbol[description={
    A vote against the proposal (``Against'').
}]{againstVote}{
    \ensuremath{\texttt{A}}
}
\newcommand{\againstVote}{{\gls[hyper=false]{againstVote}}}

\glsxtrnewsymbol[description={
    No vote.
}]{noVote}{
    \ensuremath{\textcolor{gray}{\texttt{X}}}
}
\newcommand{\noVote}{{\gls[hyper=false]{noVote}}}

\glsxtrnewsymbol[description={
    The amount of votes cast before turn $\turn$.
}]{tally}{
    \ensuremath{\theta}
}
\newcommand{\tally}{{\gls[hyper=false]{tally}}}
\newcommand{\tallyTuple}[2]{\ensuremath{\left\langle#1, #2\right\rangle}}

\glsxtrnewsymbol[description={
    An informed voter's action set, defined as $\left\{ \textit{vote early}, \textit{wait}, \textit{vote late}, \textit{abstain} \right\}$.
}]{informedActionSet}{
    \ensuremath{S}
}
\newcommand{\informedActionSet}{{\gls[hyper=false]{informedActionSet}}}

\glsxtrnewsymbol[description={
    The set of informed voters.
}]{informedVoterSet}{
    \ensuremath{V}
}
\newcommand{\informedVoterSet}{{\gls[hyper=false]{informedVoterSet}}}

\glsxtrnewsymbol[description={
    An informed voters.
}]{informedVoter}{
    \ensuremath{v}
}
\newcommand{\informedVoter}{{\gls[hyper=false]{informedVoter}}}

\glsxtrnewsymbol[description={
    An informed voter's type.
}]{type}{
    \ensuremath{\tau}
}
\newcommand{\type}{{\gls[hyper=false]{type}}}

\glsxtrnewsymbol[description={
    An informed voter's voting strategy.
}]{strategy}{
    \ensuremath{s}
}
\newcommand{\strategy}{{\gls[hyper=false]{strategy}}}

\glsxtrnewsymbol[description={
    The cost of voting.
}]{cost}{
    \ensuremath{c}
}
\newcommand{\cost}{{\gls[hyper=false]{cost}}}

\glsxtrnewsymbol[description={
    The length of the voting window, in turns.
}]{finalTurn}{
    \ensuremath{T}
}
\newcommand{\finalTurn}{{\gls[hyper=false]{finalTurn}}}

\glsxtrnewsymbol[description={
    A single voting turn.
}]{turn}{
    \ensuremath{t}
}
\newcommand{\turn}{{\gls[hyper=false]{turn}}}

\glsxtrnewsymbol[description={
    The probability of an uninformed vote at period $\turn$ is $q_\turn$.
}]{uninformed}{
    \ensuremath{q}
}
\newcommand{\uninformed}{{\gls[hyper=false]{uninformed}}}

\glsxtrnewsymbol[description={
    The utility of a voter.
}]{utility}{
    \ensuremath{U}
}
\newcommand{\utility}{{\gls[hyper=false]{utility}}}

\glsxtrnewsymbol[description={
    Information set.
}]{infoSet}{
    \ensuremath{I}
}
\newcommand{\infoSet}{{\gls[hyper=false]{infoSet}}}

\glsxtrnewsymbol[description={
    Belief system.
}]{belief}{
    \ensuremath{b}
}
\newcommand{\belief}[2]{{\gls[hyper=false]{belief}\left({#1}\vert{#2}\right)}}

\glsxtrnewsymbol[description={
    Probability of voting.
}]{voteProb}{
    \ensuremath{p}
}
\newcommand{\voteProb}[1]{{\gls[hyper=false]{voteProb}_{#1}}}

\newcommand{\Uearly}{\utility^\forVote_{\text{vote}, t_1}}
\newcommand{\Unotearly}{\utility^\forVote_{\neg{\text{vote}}, t_1}}
\newcommand{\Uwait}{\utility^\forVote_\text{wait}}

\glsxtrnewsymbol[description={
    Cost of producing an interim tally.
}]{costTally}{
    \ensuremath{C}
}

\newcommand{\qo}{\uninformed_1}
\newcommand{\qt}{\uninformed_2}
\newcommand{\po}{\voteProb{1}}
\newcommand{\ptp}{\voteProb{2,\tallyTuple{0}{0}}}
\newcommand{\ptf}{\voteProb{2,\tallyTuple{1}{0}}}
\newcommand{\pta}{\voteProb{2,\tallyTuple{0}{1}}}

\newcolumntype{L}[1]{>{\raggedright\arraybackslash}p{#1}}
\newcolumntype{C}[1]{>{\centering\arraybackslash}p{#1}}
\newcolumntype{R}[1]{>{\raggedleft\arraybackslash}p{#1}}

\ifanonymous\author{Submission \#XYZ}\else\author{
    Aviv Yaish\\
    \texttt{aviv.yaish@mail.huji.ac.il}\\
    The Hebrew University\\
    \and
    Svetlana Abramova\\
    \texttt{svetlana.abramova@uibk.ac.at}\\
    Universität Innsbruck\\
    \and
    Rainer Böhme\\
    \texttt{rainer.boehme@uibk.ac.at}\\
    Universität Innsbruck
}
\date{}
\fi
\title{Strategic Vote Timing in Online Elections With Public Tallies}

\begin{document}
\maketitle
\begin{abstract}
    We study the effect of public tallies on online elections in a setting where voting is costly and voters are allowed to strategically time their votes.
    The strategic importance of choosing \emph{when} to vote arises when votes are public, such as in online event scheduling polls (e.\,g., Doodle), or in blockchain governance mechanisms.
    In particular, there is a tension between voting early to influence future votes and waiting to observe interim results and avoid voting costs if the outcome has already been decided.

    Our study draws on empirical findings showing that ``temporal'' bandwagon effects occur when interim results are revealed to the electorate: late voters are more likely to vote for leading candidates.
    To capture this phenomenon, we analyze a novel model where the electorate consists of informed voters who have a preferred candidate, and uninformed swing voters who can be swayed according to the interim outcome at the time of voting.
    In our main results, we prove the existence of equilibria where both early and late voting occur with a positive probability, and we characterize conditions that lead to the appearance of ``last minute'' voting behavior, where all informed voters vote late.
\end{abstract}

\section{Introduction}
\label{sec:Introduction}
Many elections are conducted sequentially, where interim results are known to the electorate and can be used by voters to inform their decisions.
Given empirical work showing that voters tend to vote for the leading candidates when votes are public \cite{morton2015exit,zou2015strategic,romero2017influence,meir2020strategic,araujo2022casting}, it is natural to consider the strategic aspect of choosing \emph{when} to vote in such settings.

The power of strategic vote timing is illustrated by the commonplace show-of-hands vote: ``early bird'' voters may sway undecided voters to follow in their direction.
Furthermore, if there are costs associated with voting (e.\,g., having to commute to a distant polling station), waiting to observe interim results allows voters to save costs, if their preferred outcome appears to have garnered enough support to win.
In particular, previous work found that voting costs affect voter turnout in blockchain governance voting, where votes are irrevocable and interim results are public~\cite{dotan2023vulnerable,messias2023understanding}.
This also applies to settings in which costs may be implicit, such as in democratic deliberation dialogues \cite{flanigan2023distortion} and social networks \cite{alon2012sequential}, where voters may face social consequences if their vote does not conform to the accepted norms (e.\,g., liking a controversial social media post).

Importantly, the ability of agents to strategically time their votes may harm others, if it delays the time until an outcome can be decided.
This is significant for blockchain governance mechanisms \cite{kiayias2022sok}, in which failure to quickly converge to a decision on how to address critical flaws can result in significant losses \cite{chawla2023gauntlet,pereira2023aave,mourya2023aave,young2023aave,normandi2022messari}.
Alarmingly, the average time to reach a quorum\footnote{A \emph{quorum} is the minimal amount of votes that allows a governance proposal to pass, where proposals are dropped if less votes are cast, even if all are in favor.} may be large ($1.64$ days on average in Compound \cite{messias2023understanding}), with this risk exacerbated by mechanisms that extend voting deadlines due to last-minute voting \cite{openzeppelin2024governorpreventlatequorum}.

Although prior art examined related topics, a review of the literature identifies that an analysis of strategic vote timing in the presence of voting costs and public tallies is missing (see \cref{sec:RelatedWork,tab:RelatedWork}).
These aspects are shown to be inherently connected by both empirical and experimental studies, which find that access to interim results informs voters' actions in two crucial ways: voting ``bandwagons'' form where voters increase their support for a candidate in the lead, while voters may abstain to save costs if the results indicate that their preferred candidate will probably win or, conversely, cannot win \cite{morton2015exit,zou2015strategic,romero2017influence,meir2020strategic,araujo2022casting}.
These findings motivate our research question, the answer to which could inform the design of better voting procedures:
\begin{quote}
    \emph{To what extent do voting costs and access to interim results affect voter behavior, when voters can strategically time their votes?}
\end{quote}

\subsection{Our Contributions}
In our work, we provide an answer to this question by analyzing the timing game that arises in sequential voting procedures with public interim results, in which private-value voters incur a cost for voting, and can choose when to vote, if at all.
Thus, our main conceptual contribution is a novel model in which the strategic aspect of vote timing is informed by both voting costs and access to interim results.
Our main technical contribution is twofold.
First, we prove the existence of equilibria where voters mix between voting early and late.
While prior work on strategic vote timing by private-valued agents considers settings where mixed equilibria exist only when all voters have at least two preferred candidates (out of three) \cite{dekel2014strategic,tsang2017if}, our existence result is for voters who may have a strict preference for one candidate (out of two).
Second, we characterize several types of equilibria that are of particular relevance for the blockchain setting, as dependent on various parameters such as the cost of voting.
In total, our analysis of the strategic choice of \emph{when} to vote is a natural companion to the literature on the strategic decision of choosing \emph{who} to vote for~\cite{schoenebeck2021wisdom,han2023wisdom}, and on sequential elections with an exogenous voting order \cite{alon2012sequential,mamageishvili2023large}.

\paragraphNoSkip{Our model}
We consider competitive elections between two candidates $\forVote$ and $\againstVote$, where the electorate is divided among three equally-sized voting blocs.
Each candidate is backed by one bloc of ``informed'' voters who incur a cost of $\cost$ if they vote, and have private values that determine their utility: they get a utility of $1$ if their candidate wins, and $-1$ if the candidate loses.
Thus, the outcome critically hinges on the third bloc, which comprises of ``uninformed'' swing voters who are swayed by the interim tally at the time of voting, to the extent that they vote in favor of the leading candidate, and choose one uniformly in case there is a tie.
The elections proceed in turns, where at each turn, the interim tally (i.\,e., the number of votes cast for each candidate) is available to all agents, who can use them to inform their decisions.
In particular, our voters may account for the tally when choosing their actions at each turn.
Informed voters can choose between exercising their right to vote, and either waiting for the next turn, or abstaining if the voting window has ended.
On the other hand, uninformed voters arrive to the ballot over time according to a probability distribution $\vec{\uninformed}$.
We provide a schematic depiction of our agents' action space in \cref{fig:ActionSpace}.

\begin{figure}
    \centering
    \definecolor{informed}{RGB}{200,220,230}   
\definecolor{uninformed}{RGB}{255,228,225} 
\begin{tikzpicture}[>=stealth,x=25mm]
	\draw (-1,3.5) node [above] {\textbf{Early\strut}};
	\draw (-1,-2pt) node [below] {\scriptsize $t_1$} --++(0,4pt); 

	\draw (1,3.5) node [above] {\textbf{Late\strut}};
	\draw (1,-2pt) node [below] {\scriptsize $t_2$} --++(0,4pt); 
	
	\draw [fill,black!15] (-.25,-2.5) rectangle (.25,2.5);
	\draw (0,0) node [above] {Interim} node [below] {tally};
 
	\draw [->] (-2,0) -- (2,0) node [right] (T) {\textbf{Time\strut}};
	

	\begin{scope}[yshift=2cm]  	
	
		\draw (-3,1) node [right] {\parbox{2cm}{\raggedright \textbf{Informed}\\\emph{(Players)}}};

		\small
	
		\draw (-1,0) node [draw,circle,fill=informed] (t1) {};
		\draw (1,0) node [draw,circle,fill=informed] (t2) {};
		
		\draw [<-] (t1)--++(-3em,0) node [above] {Start\strut};
		\draw [->] (t1)--++(0,-6ex) node [below] {\parbox{3cm}{\centering Vote for\\preference}};
		
		\draw [->] (t1)-- node [above,pos=.25] {Wait\strut} (t2);
		\draw [->] (t2)--++(0,-6ex) node [below] {\parbox{3cm}{\centering Vote for\\preference}};
		\draw [->] (t2)--++(0,6ex) node [above] {Abstain\strut};
		
	\end{scope}


	\begin{scope}[yshift=-2cm]  	
	
		\draw (-3,1) node [right] {\parbox{2cm}{\raggedright \textbf{Uninformed}\\\emph{(Nature)}}};

		\small
	
		\draw (-1,0) node [draw,fill=uninformed] (t1) {};
		\draw (1,0) node [draw,fill=uninformed] (t2) {};
		
		\draw [<-] (t1)--++(-3em,0) node [above] {Start\strut};
		\draw [->] (t1)--++(0,-6ex) node [below] (Q) {\parbox{2cm}{\centering Vote\\uniformly}};
		
		\draw [->] (t1)-- node [above,pos=.25] {Wait\strut} (t2);
		\draw [->] (t2)--++(0,6ex) node [above] {Abstain\strut};
		\draw [->] (t2)--++(0,-6ex) node [below] {\parbox{2cm}{\centering Vote bandwagon$^\ast$}};
		
	\end{scope}
	
	\draw (Q.south-|T.east)++(0,-2ex) node [left] {\parbox{6cm}{\centering\scriptsize $^\ast$\,Support majority according to interim tally,\\or vote uniformly in case of ties}};
\end{tikzpicture}
    \caption{
        Schematic depiction of the action space of voters.
        By waiting, voters view the interim tally and use it when deciding their actions.
        Informed decision points are denoted by circles, while squares denote uninformed actions that are exogenously determined (i.e., whether to vote early, vote late, or abstain).
    }
    \label{fig:ActionSpace}
\end{figure}
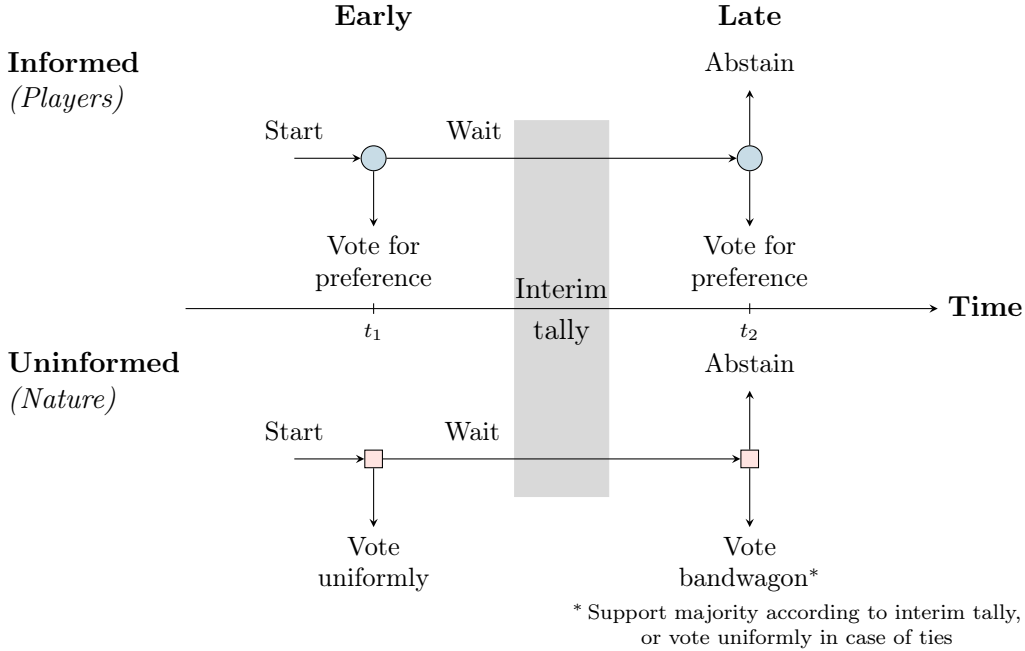

\paragraphNoSkip{Last stage analysis}
While voters have access to interim results, we show in \cref{sec:InfoSet} that this does not translate into having complete information over the state of the game.
Consider that from the point of view of a voter in favor of $\forVote$, observing an interim tally that indicates a tie between the candidates may lead to substantially different actions if all votes thus far for $\againstVote$ came from uninformed voters, while for $\forVote$ came from informed voters; in this case, the remaining informed voting power in favor of $\againstVote$ is greater than that of $\forVote$.
However, interim tallies only provide the number of votes supporting each candidate, and not the identities of the voters that cast them.
Thus, we analyze our game through the lens of the \gls{WPBE} solution concept, which lends itself to sequential games of incomplete information. It requires agent strategies to be sequentially rational, given that voters update their beliefs about the current state of the world in a Bayesian manner.
We show in \cref{prop:MustardSet,prop:GreenSet} that \gls{WPBE} strategy profiles for the last stage of our game exhibit a ``threshold'' structure: voting dominates for costs lower than the threshold and abstaining dominates above it.
Moreover, these thresholds are determined, in part, according to agent beliefs with respect to the identities of those who voted in previous turns.
For example, if an agent estimates that votes in favor of its preference were cast by uninformed voters, it may be profitable to act as a ``free rider'' and abstain (\cref{prop:WhiteSet}).

\paragraphNoSkip{Timing game analysis}
In preparation of our key technical contributions, we examine the ``timing game'' that forms when considering all voting stages (see \cref{fig:GameTree}).
Thus, in \cref{prop:EarlyUtility,prop:LateUtility,cor:EarlyWaitDiff}, we consider the impact of each game state on voter utility, when accounting for their beliefs.
These are then applied to reach our main results, where we analyze several emergent behaviors that are of importance to our setting.
For example, \citeauthor{messias2023understanding} \cite{messias2023understanding} highlighted the possibility of blockchain governance voters ``mimicking'' their peers by waiting and observing their votes.
To capture this scenario, we characterize in \cref{thm:QoZeroStrategiesTurnOne} all equilibria strategy profiles that arise when uninformed voters never vote in the first turn, but only in later stages of the voting game, i.\,e., when $\qo = 0$.
As shown in \cref{fig:QoZeroStrategies}, early voting becomes a dominant strategy.
Intuitively, this is because informed agents can, therefore, influence the votes of the uninformed agents.
We proceed with \cref{thm:WaitingEquilibria}, in which we characterize ``late-bloomer'' equilibria, where informed agents vote at the last minute.
This case is important for governance protocols, which may be required to accept or reject urgent proposals that can seriously affect the security of user funds \cite{chawla2023gauntlet,pereira2023aave,mourya2023aave,young2023aave,normandi2022messari}.
Some protocols may even delay voting deadlines in response to late-stage voting \cite{openzeppelin2024governorpreventlatequorum}, thus increasing the risk from late-bloomer equilibria.
Using our analytical characterization, the results of which are summarized in \cref{fig:EquilibriaLateBird}, we employ a numerical simulation to evaluate the average cost threshold that corresponds to the different types of such equilibria.
We find that for general values of $\qo > 0$, the minimum average threshold is $0.57$, i.\,e., more than half of the maximum utility that voters can receive from voting.
We culminate by proving the existence of equilibria where informed voters mix between voting early and late in \cref{thm:QoMixed}.

\subsection{Related Work}
\label{sec:RelatedWork}
Strategic vote timing is inherently a difficult problem in game design.
Therefore, previous work mostly suggests stylized models, backed up with solutions only for limited game settings or subsets of the strategy space.
Our review of the literature shows that strategic timing and disclosure of votes are often discussed in the context of sequential voting, where vote timing can either be \emph{exogenous} or \emph{endogenous}.
We review related works along these two scenarios and illustrate how our model differs in terms of the underlying assumptions, with a summary given in \cref{tab:RelatedWork}.
Furthermore, we present an in-depth review of the literature in \cref{sec:AdditionalRelatedWork}, which covers a variety of works, ranging from empirical blockchain papers to theoretical works on corporate governance.

\begin{table}[ht]
    \setlength{\tabcolsep}{2pt}
    \centering
    \caption{
        Overview of related works on sequential voting.
        We classify works into three categories.
        The first includes works that analyze settings where the time in which agents vote is exogenously determined (``exo.'').
        The second includes works that analyze common value elections and endogenous (``endo.'') vote timing where voters choose when to vote, thus the existence of equilibria where voters mix between voting at various times is due to the benefits of such strategies to information aggregation (``info. aggr'').
        The third includes this work, Tsang \& Larson~\cite{tsang2017if} and Dekel \& Piccione~\cite{dekel2014strategic}, that analyze endogenous vote timing in private value elections, and prove the existence of equilibria that admit mixed strategies using various modeling assumptions, as mentioned in ``Existence of mixed timing eq.'' (see \cref{sec:RelatedWork} for in-depth details).
    }
    \begin{tabular}{L{0.34\textwidth}L{0.125\textwidth}ccccL{0.2\textwidth}}
        \toprule
        \textbf{Publication}                                  & \textbf{Timing}     & \rotatebox[origin=l]{0}{\textbf{Cost}} & \textbf{Periods} & \textbf{Voters}                & \rotatebox[origin=l]{90}{\textbf{Abstention}} & \textbf{Existence of mixed timing eq.}
        \\

        \midrule

        \textbf{This work}                                    & \textbf{endo.}      & $\mathbf{\cmark}$                      & $\mathbf{2}$     & $\mathbf{3}$                   & $\mathbf{\cmark}$                             & \textbf{uninformed voters}                           \\

        Tsang \& Larson~\cite{tsang2017if}                    & endo.               & $\xmark$                               & $2$              & $3$                            & $\cmark$                                      & $2+$ preferred candidates \& \newline complete info. \\

        Dekel \& Piccione~\cite{dekel2014strategic}           & endo. with commitm. & $\xmark$                               & 2                & $N$                            & $\xmark$                                      & $2+$ preferred candidates                            \\

        \midrule

        Schmieter~\cite{schmieter2022voting}                  & endo.               & $\xmark$                               & $2$              & $N$                            & $\xmark$                                      & info. aggr.                                          \\

        Rokas \& Tripathi~\cite{rokas2007information}         & endo.               & $\xmark$                               & $\infty$         & $N \text{mod} 2=1$             & $\cmark$                                      & info. aggr.                                          \\

        \midrule

        Piketty~\cite{piketty2000voting}                      & exo.                & $\xmark$                               & 2                & $N$                            & $\cmark$                                      & $\xmark$                                             \\

        Callander~\cite{callander2007bandwagons}              & exo.                & $\xmark$                               & $\infty$         & $\infty$                       & $\xmark$                                      & $\xmark$                                             \\

        Alon et al.~\cite{alon2012sequential}                 & exo.                & $\xmark$                               & $T$              & $N$                            & $\xmark$                                      & $\xmark$                                             \\

        Battaglini et al.~\cite{battaglini2007efficiency}     & exo.                & $\cmark$                               & $T$              & $3$                            & $\cmark$                                      & $\xmark$                                             \\

        Mamageishvili \& Tejada~\cite{mamageishvili2023large} & exo.                & $\cmark$                               & $2$              & $N {\tiny \rightarrow} \infty$ & $\cmark$                                      & $\xmark$                                             \\

        \bottomrule
    \end{tabular}
    \label{tab:RelatedWork}
\end{table}

\paragraphNoSkip{Exogenous timing}
In exogenous timing voting models, agents vote at predetermined times, and cannot strategically choose when to vote.
The closest work to ours is by Mamageishvili \& Tejada~\cite{mamageishvili2023large}, who suggest a model of a large electorate divided into two groups of voters voting sequentially.
Voters choose from two alternatives, have a common voting cost, and voters in the second group know the first group's turnout, but not the interim results.
In contrast to our work, members of the first group who abstain are not allowed to vote in the second round.
The authors prove the existence of a perfect Bayesian equilibrium in which the expected number of votes for each alternative is equal.
Intuitively, this happens due to a direct effect of the early voters who vote for their preferred outcome, and an indirect effect of the late voters who consider the interim turnout when deciding whether to vote or abstain.
Battaglini et al.~\cite{battaglini2007efficiency} examine a sequential voting game with 3 voters and formulate a unique path of equilibrium in pure strategies for all voting costs.
Dekel \& Piccione~\cite{dekel2000sequential} show that information cascades, i.\,e., learning from the votes of others, do not happen in an equilibrium of binary voting schemes with exogenous timing.
We take a different design approach and intentionally introduce two types of voters, thereby allowing a potential ``learning (or influencer) effect'' to manifest.
Other works~\cite{piketty2000voting,callander2007bandwagons,alon2012sequential}, in contrast, compare simultaneous and sequential timing mechanisms in a costless setting.

\paragraphNoSkip{Endogenous timing}
Another line of theoretical research models vote timing as an endogenous factor and hence, is closer to our model.
However, most works in this stream of research seek to analyze which common-value voting schemes aggregate information more efficiently and without voting costs~\cite{rokas2007information,schmieter2022voting}.
In the private value setting, Dekel \& Piccione~\cite{dekel2014strategic} suggest a model with two time periods and a pre-commitment mechanism requiring all voters to irrevocably choose a voting period prior to election and the realization of preferences.
Tsang \& Larson~\cite{tsang2017if} analyze a 2 round costless voting game with 3 candidates and 3 voters with a preference structure following the Condorcet cycle.
They show that playing a mixed-strategy Nash equilibrium applies to a complete information game.
In contrast, our model requires strictly one alternative to be favored by informed voters and assumes no additional restrictions on a preference structure.
Furthermore, we consider a costly voting mechanism and condition the existence of a mixed-strategy equilibrium on the presence of an uninformed type of voters.

\paragraphNoSkip{Empirical \& experimental works}
Our model follows findings made by both empirical and experimental works across different voting domains, such as national elections and online voting.
Thus, Morton~\emph{et al.}~\cite{morton2015exit} find that in French presidential elections held between 1981 and 2012, the availability of exit polls at the time of voting increases the probability of voters voting for the expected winner.
Similarly, Araújo and Gatto~\cite{araujo2022casting} use data from Brazil's 2018 presidential elections to show that access to official tallies prior to voting led subsequent voters to increase their support for the frontrunning candidate, and decrease it for losing candidates.
Zou~\emph{et al.}~\cite{zou2015strategic} analyze data from $345$ thousand open Doodle polls, and find that votes positively correlate with preceding ones.
A larger data set comprising more than $1.3$ million Doodle polls is studied by Romero~\emph{et al.}~\cite{romero2017influence}, who find that early voters exert a larger influence on a poll's result than late voters.
Meir, Gal \& Tal~\cite{meir2020strategic} conduct controlled online experiments that simulate real-world voting procedures, and show that participants can be classified into several types, where some vote for the leading candidates even if they have a different preference, while others remain steadfast in voting for their preference.

In summary, our model stands out from previous studies in that it explicitly captures strategic timing and the learning effect applied to the group of uninformed voters.
Voters may prefer, in certain cases, to abstain from voting to save costs.
Furthermore, the purpose of our analysis is not to examine conditions that facilitate better information aggregation, but rather to study voting dynamics given that a certain fraction of the electorate votes ``with the herd''.

\aviv{\cite{vorobyev2021information,gersbach2021effect}}

\section{Preliminaries}
\label{sec:Preliminaries}
In this work, we focus on the importance of vote timing in swaying undecided voters when interim tallies are public.
Therefore, we consider an electorate consisting of \emph{informed} voters who have a predetermined preference for a specific outcome, and \emph{uninformed} voters who are undecided and use interim results when choosing how to vote.
To capture the strategic aspect of vote timing, we let informed voters make the strategic choice of deciding \emph{when} to vote.
We furthermore consider a setting in which voting is costly for informed voters, and uninformed voters stochastically arrive to vote over the span of the voting period.
These two aspects of our model incentivize informed voters to mix between voting early to influence future uninformed voters, and waiting to observe interim results, thereby avoiding voting costs if their vote cannot affect the final outcome.

\paragraphNoSkip{Voting mechanism}
A referendum is held to decide whether a proposal should be adopted, with the ballot open for $\finalTurn \in \mathbb{N}$ turns.
When the ballot is open, voters can vote $\forVote$ in favor of the proposal, or $\againstVote$ against it, and furthermore, voters may choose to abstain.
Votes are irrevocable and cast simultaneously within each period.
The winning outcome is chosen according to the majority vote rule after all turns.
In case of a tie, the winner is chosen by a fair coin flip.

\paragraphNoSkip{Interim tally}
Importantly, the interim results after each turn are observable before the following turn.
These results are \emph{anonymous}, i.\,e., the identity or type of those who cast previous votes is not disclosed.
We denote the tally of votes cast before turn $\turn$ by $\tally_\turn = \tallyTuple{\tally^\forVote_\turn}{\tally^\againstVote_\turn} \in \left( \mathbb{N} \cup \left\{0\right\} \right)^2$, where $\tally^\forVote_\turn$ and $\tally^\againstVote_\turn$ are the counts of $\forVote$ and $\againstVote$ votes, respectively, and define $\tally^\forVote_1 = \tally^\againstVote_1 = 0$ for the first turn.

\paragraphNoSkip{Uninformed voters}
Uninformed voters do not strictly prefer one outcome over another, but value conforming with the majority.
Thus, they cast votes in accordance to the interim results at the time of voting, and will vote for the outcome that currently enjoys a majority.
In case of a tie or when voting in the first round, uninformed voters choose to vote either $\forVote$ or $\againstVote$ with an equal probability.

\paragraphNoSkip{Uninformed arrivals}
The distribution of uninformed arrivals is characterized by a probability vector $\vec{\uninformed} = \left( \qo, \dots, \uninformed_\finalTurn \right)$, where $\uninformed_\turn \in \left[0,1\right]$ denotes the probability of arriving in round $\turn$, and $\sum_{\turn=1}^\finalTurn \uninformed_\turn \le 1$.
A strict inequality ($\sum_{\turn=1}^\finalTurn \uninformed_\turn < 1$) implies that uninformed voters may abstain.

\paragraphNoSkip{Informed voters}
Informed voters are classified into two types according to their preferred outcome.
If an informed voter is in favor of the proposal, we denote its type by $\forVote$.
Such voters receive a utility of $1$ if the proposal passes, and $-1$ otherwise.
Similarly, $\againstVote$ denotes informed voters who are against the proposal, with voters of this type gaining $-1$ if the proposal is adopted, and $1$ otherwise.
Both types incur a voting cost $\cost \in \left[0, 1\right)$, but can abstain to avoid it.

\paragraphNoSkip{Informed actions}
Immediately before the referendum begins, informed voters observe $\vec{\uninformed}$.
After the ballot opens, in each turn $\turn \in 1, \dots, \finalTurn$, informed voters who have not yet cast a vote observe the interim tally $\tally_\turn$, and then choose an action; voters may either: (1) \emph{vote} for their preference, (2) \emph{wait} until the next turn (if $\turn < \finalTurn$), or (3) \emph{abstain} (if $\turn = \finalTurn$).
As votes are irrevocable, if an informed voter has cast its vote, it does not have any action available to it in the following turns.
In particular, we are interested in instances where informed voters cast their votes in the first and last rounds, which we respectively call voting \emph{early} and \emph{late}.
We denote the set of all actions by $\informedActionSet \define \left\{ \textit{vote}, \textit{wait}, \textit{abstain} \right\}$, and the set of all distributions over these actions by $\Delta\left(\informedActionSet\right)$.

\paragraphNoSkip{Informed strategies}
Denote the set of informed voters by $\informedVoterSet$.
A strategy $\strategy$ for voter $\informedVoter \in \informedVoterSet$ is a mapping between its type and a game state to a distribution over $\informedActionSet$.
A game state is defined by a distribution of uninformed arrivals $\uninformed$, a turn index $\turn$, and the tally up to that turn $\tally_\turn = \tallyTuple{\tally^\forVote_\turn}{\tally^\againstVote_\turn}$.
Given a game state, we denote the strategy of an informed voter $\informedVoter$ of type $\type$ by $\strategy_\informedVoter \left( \uninformed, \turn, \tally^\forVote_\turn, \tally^\againstVote_\turn, \type \right)$.

\subsection{Solution Concept}
We analyze weak perfect Bayesian equilibria, which are natural for multi-turn games of incomplete information~\cite{fudenberg1991perfect}.
We focus on strategy profiles that are symmetric ``within'' and ``across'' types.

\paragraphNoSkip{Symmetry within types}
Let $\informedVoterSet_\type$ be the set of type $\type$ voters.
A strategy profile is symmetric \emph{within} $\type$ if, for any state, all voters in $\informedVoterSet_\type$ use the same strategy: $\forall \informedVoter, \informedVoter' \in \informedVoterSet: \strategy_\informedVoter \left( \uninformed, \turn, \tally^\forVote_\turn, \tally^\againstVote_\turn, \type \right) = \strategy_{\informedVoter'} \left( \uninformed, \turn, \tally^\forVote_\turn, \tally^\againstVote_\turn, \type \right)$.

\paragraphNoSkip{Symmetry across types}
A strategy profile is called symmetric \emph{across} types if for any state, for all voters, the strategy of both types is symmetric when considering the amount of votes cast by each: $\forall \tally^\forVote_\turn, \tally^\againstVote_\turn \in \mathbb{N} \cup \left\{0\right\}: \strategy_\informedVoter \left( \uninformed, \turn, \tally^\forVote_\turn, \tally^\againstVote_\turn, \forVote \right) = \strategy_\informedVoter \left( \uninformed, \turn, \tally^\againstVote_\turn, \tally^\forVote_\turn, \againstVote \right)$.

\paragraphNoSkip{Simplified notations}
For profiles that are symmetric within both types, we omit $\informedVoter$ from our notations, i.\,e., use $\strategy \left( \uninformed, \turn, \tally^\forVote_\turn, \tally^\againstVote_\turn, \type \right)$ instead of $\strategy_\informedVoter \left( \uninformed, \turn, \tally^\forVote_\turn, \tally^\againstVote_\turn, \type \right)$.
Furthermore, \gls{wlog}, for profiles that are symmetric across types, we focus our analysis on $\forVote$ type voters.
Thus, we omit voter types from our notations, and succinctly denote strategies by $\strategy \left( \uninformed, \turn, \tally^\forVote_\turn, \tally^\againstVote_\turn \right)$.

\paragraphNoSkip{Information sets}
Although voters can view the interim tally at each state, this does not always suffice for them to deduce the state they are in, because they do not know who voted.
As we show, the identity of previous voters (i.\,e., whether they are informed or not) can be crucial.
If a voter cannot distinguish between several states, we say that they belong to the same \emph{information set}.
For consistency, we use the term to also refer to sets that comprise a single state, and thus voters have complete information upon reaching them.

\paragraphNoSkip{\Gls{WPBE}}
The \gls{WPBE} solution concept allows us to reason about sequential games with incomplete information, where players update their beliefs in a Bayesian manner.
Intuitively, \gls{WPBE} requires that for states that have a positive probability of being reached, player beliefs are always consistent with Bayes' rule.
\begin{definition}[\Gls{WPBE}]
    \label{def:WPBE}
    A \gls{WPBE} is a belief system $\gls[hyper=false]{belief}$ and a strategy profile $s$ such that:
    \begin{itemize}
        \item The belief system $\belief{\cdot}{\infoSet}$ is \emph{consistent} given strategy profile $s$: it defines a conditional probability distribution over the states contained in information set $\infoSet$ for all possible information sets, and beliefs are derived using Bayes' rule for any information set that is reached with a strictly positive probability under $s$.
        \item The strategy profile $s$ is \emph{sequentially rational} given belief system $\gls[hyper=false]{belief}$, meaning that for any information set $\infoSet$, following the strategy $s(\infoSet)$ maximizes an agent's expected utility according to $\belief{\cdot}{\infoSet}$, when played from $\infoSet$ onward.
    \end{itemize}
\end{definition}

\begin{remark}
    Given an information set $\infoSet$ and state $\sigma \in \infoSet$, recall that Bayes' rule is defined as $P(\sigma \vert \infoSet) = \frac{P(\infoSet \vert \sigma) P(\sigma)}{P(\infoSet)}$, i.\,e., the probability of an agent being in state $\sigma$ given that the agent is in information set $\infoSet$, equals the conditional probability of being in information set $\infoSet$, conditioned on the probability of being in state $\sigma$, times the probability of being in $\sigma$, divided by the probability of being in information set $\infoSet$.
    As we show in \cref{sec:InfoSet}, the information set that a voter reaches at any given turn is defined according to the interim tally and the voter's action in the preceding turn.
    Thus, voters can discern to which information set they have arrived (while they cannot necessarily discern the specific state within the set).
    This implies that as $\sigma \in \infoSet$, we have $P(\infoSet \vert \sigma) = 1$, and thus we get: $P(\sigma \vert \infoSet) = \frac{P(\sigma)}{P(\infoSet)}$.
    Conversely, for $\sigma \notin \infoSet$, we get $P(\infoSet \vert \sigma) = 0$, and thus also $P(\sigma \vert \infoSet) = 0$.
\end{remark}

\subsection{Timing Game}
Due to the complexity of our game, we focus on the strategic behavior of voters facing a two-period voting window $\finalTurn \define 2$, given that the electorate consists of an uninformed voter, an informed voter in favor of $\forVote$, and an informed voter supporting $\againstVote$.
Intuitively, our setting is inherently that of imperfect but complete information.
Informed voters know the interim tally at each round, yet they do not always know \emph{who} voted in all the game's different information sets.
This uncertainty prevents voters from discerning whether they can affect the vote's outcome and introduces risk: voting in such states may lower utility.
As we show, in equilibria, rational informed voters account for the probability of reaching such states.

We proceed with an overview of our game (summarized in~\cref{fig:ActionSpace} and depicted in full in~\cref{fig:GameTree}), and then succinctly describe our approach to solving the game and finding equilibrium strategies.

\paragraphNoSkip{Overview}
The game tree presented in~\cref{fig:GameTree} takes the perspective of the informed voter $\forVote$ and accounts for the arrival of the uninformed voter in period $\turn_1$ (with probability $\qo$, see the upper part of~\cref{fig:GameTree}), $\turn_2$ (with probability $\qt$, see the medium part of~\cref{fig:GameTree}), or abstaining (with probability $q_\varnothing = (1-\qo-\qt)$, see the bottom part of~\cref{fig:GameTree}). The slope branches represent two strategies available to the informed $\forVote$-type voter: \textit{vote} and \textit{wait} in period $\turn_1$, or \textit{vote} and \textit{abstain} in the next period $\turn_2$ only in those subgames which start with the \textit{wait} strategy in $\turn_1$.
The cornered branches represent moves of other voters.
The interim and final results are specified in form of the observed tally of votes. The last two columns of the game tree decode each tuple of the tally with a sequence of the informed $\forVote$-type voter's move ($\forVote$ or $X$ in case of abstaining), the informed $\againstVote$-type voter's move, and the uninformed voter's move ($\forVote$, $\againstVote$ or $X$ depending on the  outcome of the game).
The utility for each leaf node is provided with respect to the informed  $\forVote$-type voter.

\paragraphNoSkip{Information sets}
There are several information sets within this game, in which the informed  $\forVote$-type voter is not able to distinguish between the moves of the $\againstVote$-type or the uninformed voters.
For instance, the interim tally $\tallyTuple{1}{1}$ forms a set with three states, in which the informed $\forVote$-type voter has voted, but is uncertain about the origin of the opposing vote $\againstVote$.
The information set $\tallyTuple{1}{0}$ (in \emph{blue}) models a situation, in which the informed $\forVote$-type voter has no further action available, yet is uncertain whether the opponents will cast a vote or abstain.
The information set $\tallyTuple{0}{0}$ (in \emph{mustard}) includes two states, in which the informed $\forVote$-type voter can either vote or abstain in period $t_2$, while being uncertain about the arrival of the uninformed voter.
Similarly, in the information set $\tallyTuple{0}{1}$ (in \emph{green}), the informed $\forVote$-type voter may choose from the two actions, however cannot determine whether the registered $\againstVote$-vote belongs to the informed $\againstVote$-type voter or the uninformed voter arriving and voting in period $t_1$ based on the outcome of a coin toss.
Finally, the single node $\tallyTuple{1}{0}$ (in a ``\emph{white}'' double circle) captures the state, in which the informed $\forVote$-type player observes the $\forVote$ vote of the uninformed voter.

\begin{figure}
    \centering
    \input{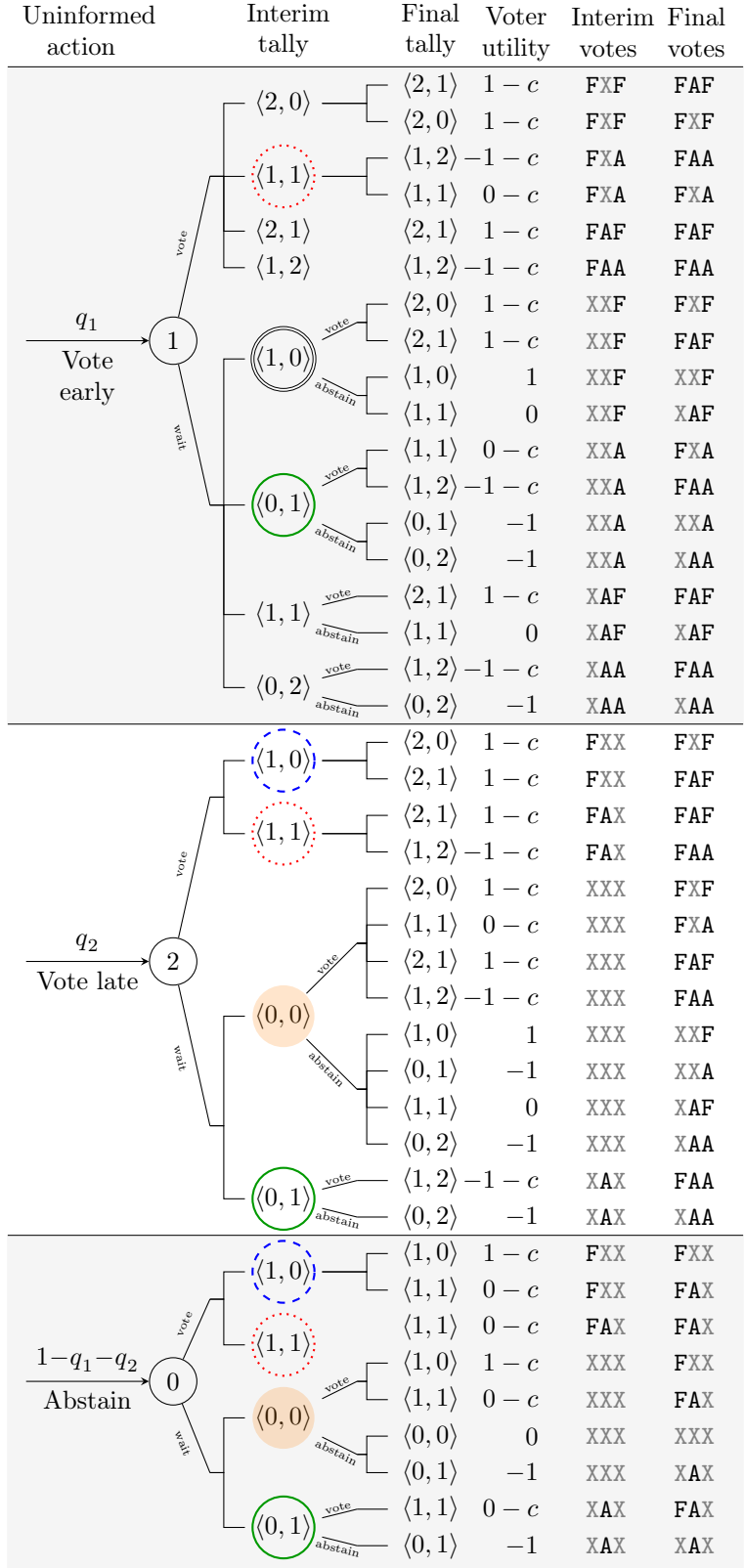}
    \caption{
        A game tree, where the voting window is two periods long and voting blocs have equal weight.
        Sloped branches are player moves; cornered branches are other voters' moves.
        Color indicates interim information sets.
    }
    \label{fig:GameTree}
\end{figure}

\paragraphNoSkip{Solving the game}
Recall that a $\forVote$ type voter can either vote early, or wait to view the interim results, and then decide whether to vote late or abstain.
From the point of view of our voter, for some possible interim result and a sequence of actions that led to this result, the states that correspond to the result are indistinguishable and thereby belong to the same information set.
As our solution concept is \gls{WPBE} (see~\cref{def:WPBE}), equilibrium strategies must be sequentially rational, and thus also optimal within each information set.
This allows us to derive strategies for the different sets that correspond to the last voting period separately, and proceed by backward induction.

\paragraphNoSkip{Notations}
Given that $\forVote$ type voters cannot distinguish between states with the same interim tally that are reached with the same action sequence, we denote the probability a $\forVote$ type voter assigns to voting early by $\po$, and the probability that the voter votes late given that the interim results were
$\infoSet \in \left\{ \tallyTuple{0}{0}, \tallyTuple{1}{0}, \tallyTuple{0}{1}, \tallyTuple{1}{1}, \tallyTuple{0}{2} \right\}$ by $\voteProb{2,\infoSet}$.

\section{Information Set Analysis}
\label{sec:InfoSet}
We proceed by analyzing each of our game's information sets from the perspective of $\forVote$ type voters.
Recall that we analyze the game via backward induction, thus we focus on second-turn information sets where agents have actions available to them.
In particular, these sets are reachable only by waiting in the first turn: if agents voted in the first turn, they do not have any action available to them until the game is over.
Among these sets, the ones that are of the most interest are those in which voters have at least two undominated strategies.
For each such set, we characterize a \emph{cost threshold}: voting dominates when voting costs are lower than the threshold, while abstaining dominates if they are higher.

\subsection{Sets With Undominated Strategies}
\paragraphNoSkip{Mustard set}
In \cref{prop:MustardSet}, we analyze the information set that encompasses all states where the interim results equal $\tallyTuple{0}{0}$ that are reachable if both informed voters wait in the first turn.
In \cref{fig:GameTree}, such states are colored in mustard.
Intuitively, \cref{prop:MustardSet} implies that in this information set, informed voters highly account for the arrival distribution of uninformed voters when choosing whether to vote or not.
Thus, informed votes would rather abstain if, given the distribution, it is unlikely that uninformed voters would turn out to vote, and would vote only if it is cheap to do so, i.e., when $\cost$ is low.
As this information set corresponds to the second turn, one can deduce that uninformed voters have a low likelihood of voting in the following cases:
\begin{itemize}
    \item If uninformed voters tend to vote early, i.e., as $\qo$ approaches one.
    \item If uninformed voters do not tend to vote in the second round, i.e., when $\qt$ approaches zero.
\end{itemize}
\begin{restatable}[]{proposition}{propMustardSet}
    \label{prop:MustardSet}
    The equilibrium strategy for states with an interim tally of $\tallyTuple{0}{0}$ is:
    \begin{equation*}
        \ptp
        =
        \begin{cases}
            1,                   & \left(\left(\qo < 1\right) \wedge \left(\cost < 1 + \frac{\qt}{2 \left(\qo - 1\right)}\right)\right)
            \\
            0,                   & \left(\left(\qo < 1\right) \wedge \left(\cost > 1 + \frac{\qt}{2 \left(\qo - 1\right)}\right)\right)
            \\
            \text{indifference}, & \left(\left(\qo < 1\right) \wedge \left(\cost = 1 + \frac{\qt}{2 \left(\qo - 1\right)}\right)\right) \vee \left(\qo = 1\right)
        \end{cases}
    \end{equation*}
\end{restatable}
We prove \cref{prop:MustardSet} in \cref{sec:Proofs}.
In the proof, we first identify the set's states, define the beliefs of agents given that a state belonging to this set is reached, and finally perform an analysis of equilibrium strategies.

\paragraphNoSkip{Green set}
In \cref{prop:GreenSet}, we analyze the information set that corresponds to states with an interim tally equal to $\tallyTuple{0}{1}$, which are colored green in \cref{fig:GameTree}.
This set is interesting, as upon reaching it, having certain knowledge of the exact state that is reached is possible only in degenerate cases.
The proof of \cref{prop:GreenSet} follows the lines of the proof of \cref{prop:MustardSet}, and is provided in \cref{sec:Proofs}.
\begin{restatable}[]{proposition}{propGreenSet}
    \label{prop:GreenSet}
    The equilibrium strategy for states with an interim tally equal to $\tallyTuple{0}{1}$ is:
    \begin{equation*}
        \pta
        =
        \begin{cases}
            1,                   & \cost < 1 + \frac{(\ptf \qo - 2 \qt) \po - \ptf \qo}{(2 - 3 \qo)\po + \qo}
            \\
            0,                   & \cost > 1 + \frac{(\ptf \qo - 2 \qt) \po - \ptf \qo}{(2 - 3 \qo)\po + \qo}
            \\
            \text{indifference}, & \cost = 1 + \frac{(\ptf \qo - 2 \qt) \po - \ptf \qo}{(2 - 3 \qo)\po + \qo}
        \end{cases}
    \end{equation*}
\end{restatable}

\paragraphNoSkip{White set}
The last information set that has at least two undominated strategies is the one that consists of all states with an interim tally of $\tallyTuple{1}{0}$ which are reachable by waiting in the first turn.
Although there is only one such state and thus there is no uncertainty for the $\forVote$-type voter upon reaching it, it proves critical for our analysis of the first stage of the game, as in this state, our voter can enjoy the ``free rider'' effect: the single vote in support of $\forVote$ is surely cast by the uninformed agent, implying that in certain cases, the $\forVote$-type voter may be able to avoid incurring the cost of voting, yet still have its preferred candidate win.
This state is denoted by two nested circles that form a single circle with a white outline in \cref{fig:GameTree}.
This information set is analyzed in \cref{prop:WhiteSet}.
As before, the corresponding proof is given in \cref{sec:Proofs}, and is similar to that of \cref{prop:MustardSet}.
\begin{restatable}[]{proposition}{propWhiteSet}
    \label{prop:WhiteSet}
    Given that the $\forVote$ type voter did not vote early and a state with interim tally equal to $\tallyTuple{1}{0}$ is reached, the equilibrium strategy is:
    \begin{equation*}
        \ptf
        =
        \begin{cases}
            1,                   & ((\qo > 0) \wedge (\cost < \pta))
            \\
            0,                   & ((\qo > 0) \wedge (\cost > \pta))
            \\
            \text{indifference}, & ((\qo > 0) \wedge (\cost = \pta)) \vee (\qo = 0)
        \end{cases}
    \end{equation*}
\end{restatable}

\subsection{Sets With Dominated Strategies or no Actions}
While the information sets with undominated strategies are of particular interest, we emphasize that the others are important for our analysis of first-turn voter strategies.
Thus, we go over them.

\paragraphNoSkip{Sets with available actions}
There are two information sets with dominated strategies in which voters have actions available to them.
In the first case, we prove in \cref{prop:ZeroTwoSet} that abstaining dominates voting for almost all $\cost$ values, while voting dominates in the second set, as shown by \cref{prop:OneOneSet}.

\begin{restatable}[]{proposition}{propZeroTwoSet}
    \label{prop:ZeroTwoSet}
    Given that the $\forVote$ type voter did not vote early and a state with interim tally equal to $\tallyTuple{0}{2}$ is reached, the equilibrium strategy is:
    $
        \voteProb{2,\tallyTuple{0}{2}}
        =
        \begin{cases}
            0,                   & \cost > 0
            \\
            \text{indifference}, & \cost = 0
        \end{cases}
    $.
\end{restatable}
\begin{proof}
    There is just one state which has a tally of $\tallyTuple{0}{2}$ and is reachable by not voting in the first turn.
    According to the interim tally, our $\forVote$-type voter cannot tilt the outcome, as there are two votes for $\againstVote$.
    Thus, voting is futile and results in a utility of $-1-\cost$.
    Abstaining gives a utility of $-1$, and thus strictly dominates voting when $\cost > 0$.
\end{proof}

\begin{restatable}[]{proposition}{propOneOneSet}
    \label{prop:OneOneSet}
    Given that the $\forVote$ type voter did not vote early and a state with interim tally equal to $\tallyTuple{1}{1}$ is reached, the equilibrium strategy is:
    $
        \voteProb{2,\tallyTuple{1}{1}}
        =
        1
    $.
\end{restatable}
\begin{proof}
    The described set contains exactly one state.
    In this state, as the $\forVote$-type voter waited in the first turn and as two other voters have voted, one can deduce that all uninformed and $\forVote$-type voters have voted.
    Thus, the $\forVote$-type voter can vote in the current turn and ensure that $\forVote$ reigns supreme, netting a utility of $1-\cost$.
    On the other hand, abstaining will result in the winner being chosen by a coin-flip, and thus in a utility equal to $0$.
    As by definition we have $\cost < 1$, this means that voting strictly dominates abstaining.
\end{proof}

\paragraphNoSkip{Sets with no available action}
There are five information sets in which voters have no action available to them.
In particular, these sets are reachable only for agents who voted in the first turn.
\begin{itemize}
    \item The set containing all states with an interim tally equal to $\tallyTuple{1}{1}$.
          We highlight these with red dotted circles in \cref{fig:GameTree}.
          The states in this set are indistinguishable for voters if the uninformed arrival probability is not entirely concentrated in a single event, i.e., if at least one of the following holds: $\qo \ne 1$, $\qt \ne 1$, or $\qo+\qt \ne 1$.
    \item The set consisting of all states that have an interim tally of $\tallyTuple{1}{0}$.
          These states are denoted in \cref{fig:GameTree} by dashed blue circles.
          We note that voters can distinguish between the set's states if either $\qt = 0$, or $\qo+\qt=1$.
    \item The states with interim tallies equal to $\tallyTuple{2}{0}$, $\tallyTuple{2}{1}$, and $\tallyTuple{1}{2}$ correspond each to a separate ``singleton'' information set comprising a single state.
          As each set contains just one state, voters have full information upon reaching them.
\end{itemize}

\section{Timing Game Analysis}
So far, we have analyzed information sets that are reachable only if the $\forVote$ voter did not vote early.
We proceed to examine the ``timing game'' that arises when voters strategically choose when to vote, including the possibility of voting early.

\subsection{Utility of Voting Early and Waiting}
We first write the expected utility that can be gained from either voting in the first turn, or waiting to the second one.
%
We denote the expected utility of the informed $\forVote$-type voter from voting early by $\Uearly$ and from waiting by $\Uwait$.
In \cref{prop:EarlyUtility,prop:LateUtility} we explicitly derive these expected utilities.
Briefly, both values are obtained by analyzing all possible outcomes for following the respective strategy (i.e., either voting early or waiting), and summing the resulting utilities.
The full proofs are given in \cref{sec:Proofs}.

\begin{restatable}[]{proposition}{propEarlyUtility}
    \label{prop:EarlyUtility}
    The utility gained from voting early is:
    \begin{equation*}
        \Uearly =
        (1 - \po) \cdot (1 - \qo + (1 - \qo - \qt)\cdot\pta) - \cost
    \end{equation*}
\end{restatable}

\begin{restatable}[]{proposition}{propLateUtility}
    \label{prop:LateUtility}
    The utility gained from waiting in the first turn is:
    \begin{align*}
        \Unotearly
         & =
        (\po-1) \cost \ptp
        -
        \po (
        1 + (\cost-1+\qt) \pta
        )
        \\
         & -
        \frac{\qo}{2} (
        (\cost-2) \po (1 - \pta)
        -
        \cost (
        (\po - 1) (
        2\ptp
        -
        \ptf
        -
        \pta
        )
        -
        \po\pta
        )
        )
    \end{align*}
\end{restatable}

Our analysis relies on comparing the utility of voting early with that of waiting.
To this end, we present two useful results in \cref{cor:EarlyWaitDiff,cor:BothIndifferencePo}, both of which are proven in \cref{sec:Proofs}.
\begin{restatable}[]{corollary}{corEarlyWaitDiff}
    \label{cor:EarlyWaitDiff}
    The difference in utility between voting early and voting late is:
    \begin{align*}
        \Uearly - \Unotearly
         & =
        \left(
        \ptp - 1 +
        \left(\frac{1}{2}  \left(\pta + \ptf\right) - \ptp\right) \qo
        \right) \cost
        \nonumber \\&
        - \left(
        \left(3 \pta + \ptf - 2 \ptp - 1\right) \frac{\qo}{2}
        + \left(\ptp - \pta\right)
        \right) \cost \po
        \nonumber \\&
        + 1 - \qo
        - \left(1 - \qo - \qt\right) \pta
    \end{align*}
\end{restatable}

\begin{restatable}[]{corollary}{corBothIndifferencePo}
    \label{cor:BothIndifferencePo}
    In an equilibrium strategy profile where voters are indifferent between voting early and waiting, we have:
    \begin{equation*}
        \po
        =
        -\frac{
            (\ptp(\qo - 1) - \frac{1}{2}(\pta + \ptf)\qo + 1)\cost + \pta(1 - \qo - \qt) + \qo - 1
        }{
            (\frac{1}{2}(3\pta + \ptf - 2\ptp - 1)\qo - \pta + \ptp)\cost
        }
    \end{equation*}
\end{restatable}

\subsection{Equilibria Analysis}
We proceed by analyzing the different possible equilibria for our timing game.
For large regimes of our parameter space, we provide a precise characterization.
Several types of equilibria are of particular interest to us:
\begin{itemize}
    \item \citeauthor{messias2023understanding} \cite{messias2023understanding} conduct an empirical analysis of blockchain governance protocols, and identify a potential risk in voters ``mimicking'' their peers by waiting to observe their votes.
          We analyze this case by setting $\qo = 0$, implying that uninformed agents never vote early.
          We characterize all equilibria in \cref{thm:QoZeroStrategiesTurnOne}, with a summary given in \cref{fig:EquilibriaLateBird}.
    \item We continue by characterizing all equilibria where informed voters vote late in \cref{thm:WaitingEquilibria}.
          This is important for governance protocols, as they may extend voting delays as a response to late voting, and this may adversely affect the safety of user funds, if the proposals voted on are related to security issues \cite{chawla2023gauntlet,pereira2023aave,mourya2023aave,young2023aave,normandi2022messari}.
    \item Finally, we prove the existence of equilibria strategies that mix early and late voting in \cref{thm:QoMixed}.
\end{itemize}

We begin by fully characterizing the case where $\qo = 0$ in \cref{prop:QoZeroStrategiesTurnTwo,thm:QoZeroStrategiesTurnOne}, with a schematic illustration of both results provided in \cref{fig:QoZeroStrategies}.
We find that in this case, second-turn strategies reduce to a clean characterization.
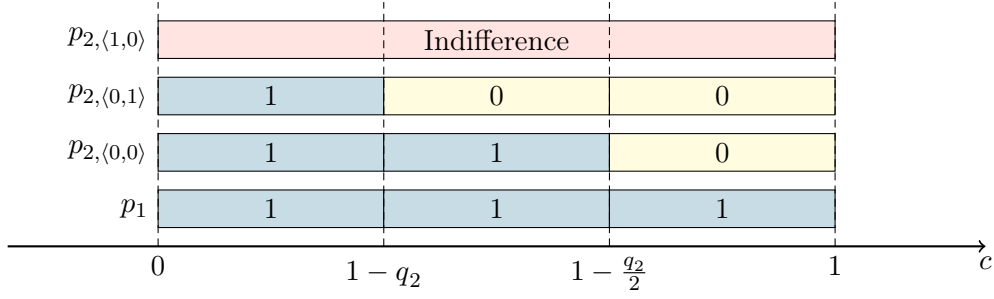
\begin{figure}
    \centering
    \definecolor{informed}{RGB}{200,220,230}   
\definecolor{uninformed}{RGB}{255,228,225} 
\begin{tikzpicture}
    \node[anchor=east] at  (0,2.75) {$\ptf$};
    \draw[fill=uninformed] (0,2.50) rectangle ++(9,0.5) node[pos=0.5] {Indifference};
    
    \node[anchor=east] at (0,2.00) {$\pta$};
    \draw[fill=informed]  (0,1.75) rectangle ++(3,0.5) node[pos=0.5] {1};
    \draw[fill=yellow!15] (3,1.75) rectangle ++(3,0.5) node[pos=0.5] {0};
    \draw[fill=yellow!15] (6,1.75) rectangle ++(3,0.5) node[pos=0.5] {0};

    \node[anchor=east] at (0,1.25) {$\ptp$};
    \draw[fill=informed]  (0,1.00) rectangle ++(3,0.5) node[pos=0.5] {1};
    \draw[fill=informed]  (3,1.00) rectangle ++(3,0.5) node[pos=0.5] {1};
    \draw[fill=yellow!15] (6,1.00) rectangle ++(3,0.5) node[pos=0.5] {0};

    \node[anchor=east] at (0,0.50) {$\po$};
    \draw[fill=informed]  (0,0.25) rectangle ++(3,0.5) node[pos=0.5] {1};
    \draw[fill=informed]  (3,0.25) rectangle ++(3,0.5) node[pos=0.5] {1};
    \draw[fill=informed]  (6,0.25) rectangle ++(3,0.5) node[pos=0.5] {1};
    
    \draw[draw,thick,-to] (-2,0) -- (11,0) node[below,anchor=north] {$\cost$};
    \draw[densely dashed] (0,0) node[below] {0}                   -- ++(0,3.25);
    \draw[densely dashed] (3,0) node[below,yshift=-2pt] {$1-\qt$} -- ++(0,3.25);
    \draw[densely dashed] (6,0) node[below] {$1-\frac{\qt}{2}$}   -- ++(0,3.25);
    \draw[densely dashed] (9,0) node[below] {$1$}                 -- ++(0,3.25);
\end{tikzpicture}
    \caption{
        A schematic depiction of equilibria strategies when $\qo = 0$, as given by \cref{thm:QoZeroStrategiesTurnOne,prop:QoZeroStrategiesTurnTwo}.
    }
    \label{fig:QoZeroStrategies}
\end{figure}
\begin{restatable}[]{proposition}{propQoZeroStrategiesTurnTwo}
        \label{prop:QoZeroStrategiesTurnTwo}
        When $\qo = 0$, second turn strategy profiles are determined solely according to the exogenous parameters $\cost$ and $\vec{\uninformed}$.
        Moreover, we have that voters are indifferent \gls{wrt} $\ptf$, and that:
        \begin{equation*}
            \pta
            =
            \begin{cases}
                1,                   & \cost < 1 - \qt
                \\
                0,                   & \cost > 1 - \qt
                \\
                \text{indifference}, & \cost = 1 - \qt
            \end{cases}
            ,
            \qquad
            \ptp
            =
            \begin{cases}
                1,                   & \cost < 1 - \frac{\qt}{2}
                \\
                0,                   & \cost > 1 - \frac{\qt}{2}
                \\
                \text{indifference}, & \cost = 1 - \frac{\qt}{2}
            \end{cases}
        \end{equation*}
    \end{restatable}
\begin{restatable}[]{theorem}{thmQoZeroStrategiesTurnOne}
    \label{thm:QoZeroStrategiesTurnOne}
    When $\qo = 0$, $\qt > 0$ and for $\cost \in \left[0,1-\qt\right)$, the equilibrium strategy profile is: $\po = 1$, $\ptp = 1$, and $\pta = 1$, while for $\cost \in \left(1-\frac{\qt}{2},1\right]$ we have: $\po = 1$, $\ptp = 0$, $\pta = 0$.
    In both cases, voters are indifferent \gls{wrt} $\po$.
    When $\qo = 0$, $\qt = 0$ and for $\cost = 1-\qt$, we have that $\po = 1$ and voters are indifferent \gls{wrt} $\pta$, while for $\cost = 1-\frac{\qt}{2}$, we again have $\po = 1$ and that voters are indifferent \gls{wrt} $\ptp$.
    In all cases, voters are indifferent \gls{wrt} $\ptf$.
\end{restatable}

When $\qo > 0$, the relationship between $\pta$ and $\ptf$ becomes more intricate and requires delicate analysis.
In \cref{thm:WaitingEquilibria}, we prove the existence of equilibria that admit strategy profiles where waiting dominates voting early, which we refer to as ``late-bloomer'' equilibria.
Our proof is given in \cref{sec:Proofs}.
It is constructive, and proceeds by increasingly refining our parameter space until obtaining a characterization of the parameter regimes that correspond to the examined type of equilibria.
In \cref{fig:EquilibriaLateBird}, we present a summary of our characterization.
\begin{figure}
    \centering
    \includegraphics[width=\columnwidth]{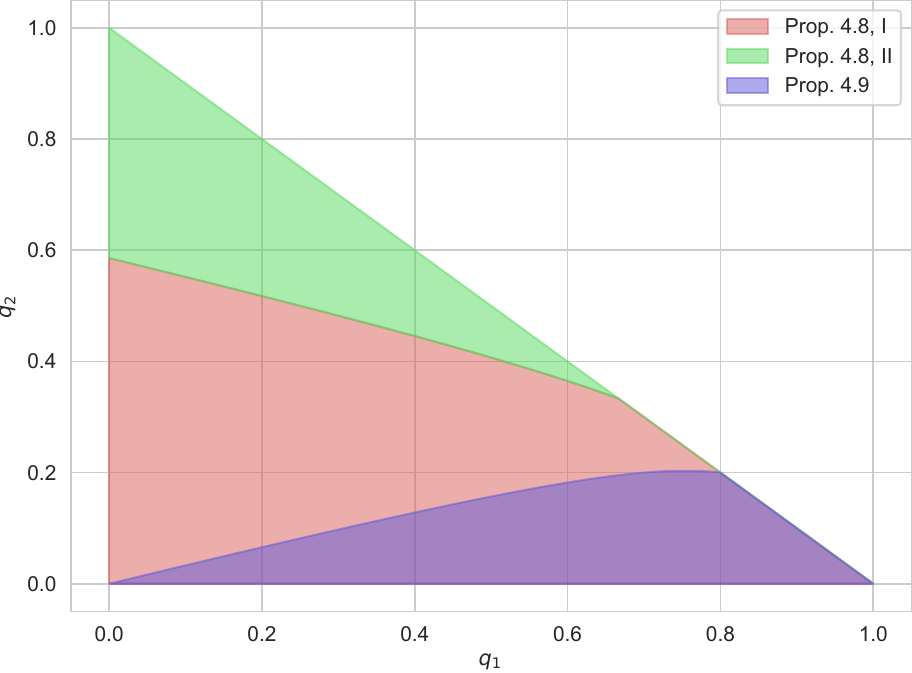}
    \caption{
        Characterization of parameter regimes where waiting dominates early voting.
        In the proof of \cref{thm:WaitingEquilibria}, we show that the cost threshold that admits such equilibria is lower when uninformed voters have a high probability $\qo$ of voting early (\cref{prop:WaitingPtpOne}), while costs are higher when this probability is low (\cref{prop:WaitingPtpZero}, which characterizes two complementary parameter regimes).
        Intuitively, this is due to the inability of informed voters to influence the uninformed decision.
        In particular, for the lower cost regime of \cref{prop:WaitingPtpOne}, we have that $\ptp = 1$, and for the higher cost regime of \cref{prop:WaitingPtpZero} has $\ptp = 0$.
        Via numeric evaluations, we find that the average threshold cost for the former is $0.57$, and for the latter is $0.76$.
    }
    \label{fig:EquilibriaLateBird}
\end{figure}

\begin{restatable}[]{theorem}{thmWaitingEquilibria}
    \label{thm:WaitingEquilibria}
    If $\qo > 0$, there are feasible parameter regimes where waiting dominates voting early.
    Moreover, equilibria strategy profiles have:
    $\po = 0$, $\pta = \cost$ and $\ptf = 1 - \cost$, while the value of $\ptp$ varies according to the parameter regime, where $\ptp = 0$ in the high cost regime (\cref{prop:WaitingPtpZero}), and $\ptp = 1$ in the low cost regime (\cref{prop:WaitingPtpOne}).
\end{restatable}
\begin{restatable}[]{proposition}{propWaitingPtpZero}
    \label{prop:WaitingPtpZero}
    In equilibria strategy profiles where waiting dominates early voting, if $\qo > 0$ and $\ptp = 0$, then either $\qo \in \left(0, \frac{3-3.5\qt - \sqrt{8.25\qt^2 -5\qt + 1}}{2}\right)$,
    $\qt \in [0, 2 - \sqrt{2})$, $\cost \in \left[\frac{2\qo + \qt - 2}{2 \left(\qo - 1\right)}, 1\right)$, or $\qo \in \left[\frac{3-3.5\qt - \sqrt{8.25\qt^2 -5\qt + 1}}{2}, \frac{3-3.5\qt + \sqrt{8.25\qt^2 -5\qt + 1}}{2}\right]$, $\qt \in \left[0,1\right]$, $\cost \in \left[\frac{\qo - 1}{1.5\qo + \qt - 2}, 1\right)$.
\end{restatable}
\begin{restatable}[]{proposition}{propWaitingPtpOne}
    \label{prop:WaitingPtpOne}
    In equilibria strategy profiles where waiting dominates early voting, if $\qo > 0$ and $\ptp = 1$, then $\qo \in \left(\frac{5\qt + 2 - \sqrt{41\qt^2 - 28\qt + 4}}{4}, \frac{5\qt + 2 + \sqrt{41\qt^2 - 28\qt + 4}}{4}\right)$, $\qt \in \left[0, \frac{14 - 4\sqrt{2}}{41}\right)$, $\cost \in \left(\frac{1 - \qo}{1 - \frac{\qo}{2} - \qt}, \frac{2\qo + \qt - 2}{2\left(\qo-1\right)}\right)$.
\end{restatable}
Intuitively, the above results imply the following:
\begin{itemize}
    \item In \cref{prop:WaitingPtpZero}, we characterize strategy profiles for the ``high cost regime'', i.e., when $\cost$ is relatively large.
          We show that in such cases, voters tend to leave the outcome to chance due to the high voting costs, which harm the utility they may gain by voting.
          Thus, informed voters do not vote in the first turn, and if they observe that noone else has voted, they abstain altogether.
          If there is a single vote against their preference, they vote with a high probability; in this case, the opposing vote presents itself as an immediate threat, and countering it by voting is worth the cost.
          Via numerical simulations, we find that the average cost threshold for this regime is $0.76$.
    \item For lower costs, we characterize corresponding strategy profiles in \cref{prop:WaitingPtpOne}, and show that the lower costs lead voters to slightly change their behavior.
          Specifically, they always vote upon reaching the so-called ``mustard'' information set, which has a tally of $\tallyTuple{0}{0}$.
          Our characterization implies that such equilibria correspond to parameter regimes where voting costs are lower and $\qo$ is significantly larger than $\qt$.
          Thus, upon reaching the second turn and observing that no votes were cast, it is likely that the uninformed agent would abstain.
          By using numerical simulations, we reach an average cost threshold of $0.57$ for this case.
\end{itemize}

Finally, we show in \cref{thm:QoMixed} that there are equilibria in which voters mix between voting early and late.
The proof is given in \cref{sec:Proofs}.
In it, we use second-stage equilibria strategies to gradually confine $\po$ into a restricted parameter regime.
Then, we culminate by showing that the resulting parameter ranges are non-empty.
\begin{restatable}[]{theorem}{thmQoMixed}
    \label{thm:QoMixed}
    When $\qo > 0$, there are equilibria strategy profiles where $\qo \in \left(0,1\right)$.
\end{restatable}

\section{Conclusion}
\label{sec:Conclusion}
We propose a novel model of strategic voting timing which differs from prior works in several important aspects.
First, we introduce ``informed'' voters who have a strict preference over the binary outcome, and ``uninformed'' voters who can be swayed by the interim public tally at the time of voting.
Second, we study the effect of voting costs on voters' equilibrium strategies, and how different parameter regimes may incentivize them to either vote early, vote late, or abstain.
Our theoretical results show that informed voters will choose a certain equilibrium path of voting or mix between voting early and late as a function of the cost of voting, and the distribution of the time at which uninformed voters arrive to the ballot.
This could be of interest and practical relevance to both voters and operators of such elections.

\section*{Acknowledgement}
The research leading to this paper was funded by the Austrian Research Promotion Agency (FFG), the Austrian Security Research Programme (KIRAS), and the Austrian Blockchain Center (ABC).
The first author is grateful for having received travel grants from the Austria-Israel Academic Network Innsbruck (AIANI) and from the Hebrew University's Unit for Diversity \& Inclusion, which made this collaboration possible.

\printbibliography
\newpage

\appendix

\section{Proofs}
\label[appendix]{sec:Proofs}
This section contains all proofs that were omitted from the body of the paper due to space constraints.

\propMustardSet*
\begin{proof}
    We proceed by examining the set's states, the beliefs of agents given that a state belonging to this set is reached, and finally with an analysis of equilibrium strategies.

    \paragraphNoSkip{States}
    This set contains states where the interim tally after the first turn is $0$ for both alternatives, implying that no-one voted early, and that either the uninformed voter will vote in the second round, or will abstain.

    \paragraphNoSkip{Beliefs}
    Game states corresponding to this information set are reachable only if both the uninformed and informed voters do not always vote early, equivalent in our notations to $\qo < 1$ and $\po < 1$.
    In particular, the set contains a single state if $\qt = 1$ or $\qo = \qt = 0$; otherwise, it contains two states.
    Given informed voters reach this set, they are in the state that corresponds to the uninformed voter voting late with probability
    \begin{equation}
        P(\text{uninformed votes late} \mid \text{informed in mustard set})
        =
        \frac{\qt}{\qt + \left( 1 - \qo - \qt \right)}
        =
        \frac{\qt}{1 - \qo},
    \end{equation}
    and, similarly, are in the state corresponding to the uninformed voter abstaining with probability
    \begin{equation}
        P(\text{uninformed abstains} \mid \text{informed in mustard set})
        =
        \frac{1-\qo-\qt}{1 - \qo}
        =
        1 - \frac{\qt}{1 - \qo}.
    \end{equation}

    \paragraphNoSkip{Equilibrium analysis}
    Recall that this information set contains at most two states, and that under our notations, the probability assigned to voting in this set is denoted by $\ptp$.
    In the state that corresponds to the uninformed voter abstaining, the game becomes a prisoner's dilemma:
    \begin{itemize}[leftmargin=*]
        \item First, we examine the cases where the $\forVote$-type voter votes.
              With probability $\ptp$, the other informed voter also votes, thus we have a tie and a winner is chosen at random: $\forVote$ wins with probability $\frac{1}{2}$ (giving a utility of $1-\cost$), and $\againstVote$ wins with the same probability (netting a utility of $-1-\cost$), resulting in an expected utility equal to $\frac{1}{2} \cdot \left(1-\cost -1-\cost\right) = -\cost$.
              Conversely, the opponent does not vote with probability $1-\ptp$, in which case $\forVote$ wins and a utility of $1-\cost$ is gained.
              The total expected utility of voting is:
              $
                  \left(1 - \ptp\right)\cdot\left(1-\cost\right)
                  +
                  \ptp \cdot \left(-\cost\right)
                  =
                  1 - \cost  - \ptp
              $.
        \item If the $\forVote$-type voter abstains, then if the informed opponent abstains (with probability $1-\ptp$), a utility of $0$ is gained as a winner is chosen at random.
              On the other hand, the opponent votes with probability $\ptp$, resulting in a utility equal to $-1$.
              Thus, the exepected utility in this case equals:
              $
                  \left(1 - \ptp\right)\cdot0
                  +
                  \ptp \cdot \left(-1\right)
                  =
                  - \ptp
              $.
    \end{itemize}
    On the other hand, if the uninformed voter does not abstain, then it will randomly choose who to vote for in the second period, as no votes were cast for both options at the time of voting.
    \begin{itemize}[leftmargin=*]
        \item Consider the case where the $\forVote$-type voter votes.
              The $\againstVote$-type voter abstains with probability $1-\ptp$, in which case a utility of $1-\cost$ is obtained if the uninformed voter's random decision results in voting for $\forVote$, and a utility of $-\cost$ is gained if the uninformed voter cast a vote for $\againstVote$.
              As the vote of the uninformed voter is decided with a uniform probability, the expected utility for this subcase is $\left(1-\ptp\right) \cdot \left(\frac{1}{2}-\cost\right)$.
              With probability $\ptp$, the $\againstVote$-type voter does vote and we obtain an expected utility of $\ptp \left(-\cost\right)$: if the uninformed voter chose $\forVote$ the utility is $1-\cost$, while it is $-1-\cost$ if the uninformed votes for $\againstVote$.
              In total, the expected utility of both subcases is:
              $
                  \left(1 - \ptp\right)\cdot\left(\frac{1}{2}-\cost\right)
                  +
                  \ptp \cdot \left(-\cost\right)
                  =
                  \frac{1}{2}-\cost-\frac{\ptp}{2}
              $.
        \item In case the $\forVote$-type voter abstains, then with probability $1 - \ptp$ the $\againstVote$-type voter abstains as well, thus the outcome depends solely on the uninformed voter's random choice, resulting in an expected utility of $0$.
              With probability $\ptp$, the $\againstVote$-type voter votes.
              In this sub-case, the expected utility is $-\frac{1}{2}$: if the uninformed voter voted for $\forVote$, then we have a tie and the utility is $0$, while if the uninformed vote was for $\againstVote$ then the utility equals $-1$.
              In total, the expected utility equals:
              $
                  \left(1 - \ptp\right)\cdot 0
                  +
                  \ptp \cdot \left(-\frac{1}{2}\right)
                  =
                  -\frac{\ptp}{2}
              $.
    \end{itemize}

    In sequentially rational equilibria strategies that involve both voting and abstaining, voters should be indifferent between the two actions in this information set's states, given their beliefs on the probability of reaching each state.
    Thus, the expected utility should satisfy the following.
    \begin{multline}
        \label{eq:PrisonIndifference}
        \overbrace{\frac{\qt}{1 - \qo} \cdot \left(
            \frac{1}{2}-\cost-\frac{\ptp}{2}
            \right)
            +
            \left( 1 - \frac{\qt}{1 - \qo} \right) \cdot \left(
            1 - \cost  - \ptp
            \right)}^{\text{expected utility of informed voting late}}
        =
        \\
        \underbrace{\frac{\qt}{1 - \qo} \cdot \left(
            -\frac{\ptp}{2}
            \right)
            +
            \left( 1 - \frac{\qt}{1 - \qo} \right) \cdot \left(
            - \ptp
            \right)}_{\text{expected utility of informed abstaining}}
    \end{multline}
    When \cref{eq:PrisonIndifference} is solved, we obtain: $\cost = 1 + \frac{\qt}{2 \left( \qo - 1\right)}$.
    Furthermore, substituting this result in the expected utility reveals that abstention is dominant in this information set when $\cost > 1 + \frac{\qt}{2 \left( \qo - 1\right)}$, and voting is dominant when $\cost < 1 + \frac{\qt}{2 \left( \qo - 1\right)}$.
    This allows us to write the equilibrium strategy as a ``threshold'' strategy that depends on $\cost$, as presented in the statement of \cref{prop:MustardSet}.
\end{proof}

\propGreenSet*
\begin{proof}
    We proceed similarly to the proof of \cref{prop:MustardSet}.

    \paragraphNoSkip{States}
    This set contains states with an interim tally of $\pta = 1$ where the only vote that was cast in the first period was for $\againstVote$.

    \paragraphNoSkip{Beliefs}
    This information set is distinguished from other second-stage sets because $\forVote$-type agents cannot always discern whether the vote for $\againstVote$ was cast by another informed agent, or by an uninformed one.
    There are three states contained within this set, all of which are reachable only if informed voters do not surely vote early, i.e., $\po < 1$.
    The first state belonging to this set is reached if uninformed agents vote early ($\qo > 0$) for $\againstVote$, and if informed agents do not always vote early ($\po < 1$).
    By definition, if the uninformed indeed voted early, the choice between $\forVote$ and $\againstVote$ is determined by a coin toss and thus has probability $\frac{1}{2}$, thus the probability of reaching this state equals $\left(1-\po\right)\cdot\qo\cdot\frac{1}{2}$.
    The second state is reachable when the uninformed votes in the second round ($\qt > 0$) and informed voters assign non-zero probability to voting early ($\po > 0$).
    It has a probability of $\po\cdot\qt$ of being reached.
    Finally, the third state is reached with probability $\po\cdot\left(1-\qo-\qt\right)$.
    It can only be reached when $\qo + \qt < 1$, and if both the uninformed abstains, and informed voters sometime vote early ($\po > 0$).
    Thus, the probability of reaching any of these states is:
    \begin{equation}
        \label{eq:GreenTotalProbability}
        \left(\left(1-\po\right) \frac{\qo}{2}\right)
        +
        \po \qt
        +
        \left(\po \left(1-\qo-\qt\right)\right)
        =
        \left(\left(1-\po\right) \frac{\qo}{2}\right)
        +
        \left(\po \left(1-\qo\right)\right)
        =
        \frac{1}{2} \left( \qo + \po \left( 2 - 3 \qo \right) \right)
    \end{equation}
    Given \cref{eq:GreenTotalProbability}, we calculate the conditional probability of being in each of the set's states.
    The probability of being in the first state, conditioned on reaching a state belonging to this set, is:
    $\frac{\left(1-\po\right) \qo}{\qo + \po ( 2 - 3\qo )}$.
    Similarly, the conditional probability of being in the second state equals:
    $\frac{2\po\qt}{\qo + \po (2 - 3\qo)}$.
    Lastly, the conditional probability of being in the third state is:
    $\frac{2\po\left(1-\qo-\qt\right)}{\qo + \po (2 - 3\qo)}$.

    \paragraphNoSkip{Equilibrium analysis}
    In all states belonging to this set, abstention results in a utility of $-1$ for the $\forVote$ type voter.
    The utility of voting varies among the states.
    In case the uninformed voter abstains, then voting results in a utility of $-\cost$.
    If the uninformed votes late, the utility from voting is $-1-\cost$.
    Finally, if the uninformed votes early, then the utility equals $-\cost$ if the $\againstVote$ type voter did not vote, and equals $-1-\cost$ if it has voted.
    Due to symmetry, $\againstVote$ voters vote in this state with probability $\ptf$.

    Thus, in a sequentially rational equilibrium, $\forVote$ voters are indifferent if the expected utility of voting equals that of abstaining, given the voter's beliefs.
    In \cref{eq:GreenIndifferenceUniformedEarly}, we derive the corresponding indifference condition.
    \begin{align}
        \label{eq:GreenIndifferenceUniformedEarly}
        (1-\po) \cdot \frac{\qo}{2}\cdot(
        ((1-\ptf) \cdot (-\cost))
         &
        +
        (\ptf \cdot (-1-\cost))
        )
        \nonumber \\&
        +
        (\po \cdot \qt \cdot (-1-\cost))
        +
        (\po \cdot (1-\qo-\qt) \cdot (-\cost))
        \nonumber \\&
        =
        (1-\po) \cdot \frac{\qo}{2} \cdot (-1)
        + (\po \cdot \qt \cdot (-1))
        \nonumber \\&
        + (\po \cdot (1-\qo-\qt) \cdot (-1))
    \end{align}

    We solve \cref{eq:GreenIndifferenceUniformedEarly} for $\cost$, giving us:
    $c = 1 + \frac{
            (
            \ptf \qo
            -
            2 \qt
            ) \po
            -
            \ptf \qo
        }{(2 - 3 \qo)\po + \qo}$.
    Moreover, we deduce from \cref{eq:GreenIndifferenceUniformedEarly} that for any cost above this value, the expected utility of abstaining is greater than that of voting, while for lower values, the opposite holds.
    As before, this insight allows us to present the equilibrium strategy as following a ``threshold'' structure in \cref{prop:GreenSet}.
\end{proof}

\propWhiteSet*
\begin{proof}
    Similarly to the analysis of the preceding information sets, we go over the states comprising this set, the beliefs of agents upon reaching such a state, and agents' equilibrium strategies.

    \paragraphNoSkip{States}
    This information set corresponds to a single state, where the uninformed voted early for $\forVote$, and none of the informed voters voted in the first round.

    \paragraphNoSkip{Beliefs}
    This state is reachable only if there is both some probability that an uninformed vote will be cast early ($\qo > 0$), and a non-zero probability of informed voters waiting in the first round ($\po < 1$).
    From the perspective of the $\forVote$ type voter, this information set contains just a single state: as this voter did not vote in the first round, and is the only informed voter favoring $\forVote$, reaching this state implies that the uninformed voter cast $\forVote$.
    For the $\againstVote$ voter, this state is indistinguishable from all other states where the interim results correspond to a single vote for $\forVote$, i.e., when reaching this state, the $\againstVote$ voter's information set is symmetric to the green information set.

    \paragraphNoSkip{Equilibrium analysis}
    Given this state is reached, the $\forVote$ type voter can vote to assure its preference will win whether the $\againstVote$ type voter turns out or not, thereby receiving a utility of $1-\cost$.
    Conversely, consider the case where the $\forVote$ type voter abstains.
    In a symmetric sequentially rational equilibrium, the $\againstVote$ type voter abstains as well with probability $1-\pta$, resulting in a win for $\forVote$ and a utility of $1$.
    With the complementary probability $\pta$, the $\againstVote$ type informed voter casts a vote, in which case the winner is chosen at random and the expected utility is $0$.
    Thus, the indifference condition for the $\forVote$ voter is:
    $1-\cost = \left(1-\pta\right) \cdot 1 + \pta \cdot 0$,
    which can be simplified to $\pta = \cost$.

    Similarly to before, our result implies that the equilibrium strategy has a cost threshold; i.e., if $\cost > \pta$ then abstaining dominates, while for lower costs, voters would prefer to vote.
    We provide the threshold strategy in the statement of \cref{prop:WhiteSet}.
\end{proof}

\propEarlyUtility*
\begin{proof}
    The expected utility is composed of contributions from the relevant and mutually exclusive branches of the game tree.
    We replace the lower indexes of each contribution term with a corresponding interim tally and the distribution $q$ to explicitly specify those branches.
    If the informed $\forVote$-type voter votes early, the following interim tallies contribute to the expected utility: $\langle2,0\rangle$, $\langle1,1\rangle$, $\langle2,1\rangle$, $\langle1,2\rangle$, and $\langle1,0\rangle$.
    The utility $\Uearly$ can be specified as:
    \begin{align}
        \Uearly                                    & =
        \utility_{q_1,\langle2,0\rangle} +
        \utility_{q_1,\langle1,1\rangle} +
        \utility_{q_1,\langle2,1\rangle} +
        \utility_{q_1,\langle1,2\rangle} +
        \utility_{q_2,\langle1,0\rangle} +
        \utility_{q_2,\langle1,1\rangle} +
        \utility_{q_\varnothing,\langle1,0\rangle} +
        \utility_{q_\varnothing,\langle1,1\rangle},                                                 \\
        \intertext{where:}
        \utility_{q_1,\langle2,0\rangle}           & =  \frac{\qo}{2} \cdot (1-\po) \cdot (1-\cost) \\
        \intertext{(because the uninformed voter votes  for $\forVote$ with probability $\frac{\qo}{2}$ and the informed $\againstVote$-type voter does not vote early with probability $(1-\po)$),}
        \utility_{q_1,\langle1,1\rangle}           & = \frac{\qo}{2} \cdot (1-\po) \cdot (-1-\cost) \\
        \intertext{(because the uninformed voter votes  for $\againstVote$ with probability $\frac{\qo}{2}$, the informed $\againstVote$-type voter does not vote early with probability $(1-\po)$, and has a dominant strategy to vote for $\againstVote$ in $t_2$),}
        \utility_{q_1,\langle2,1\rangle}           & =    \frac{\qo}{2} \cdot \po \cdot (1-\cost)   \\
        \intertext{(because the uninformed voter votes  for $\forVote$ with probability $\frac{\qo}{2}$ and the informed $\againstVote$-type voter votes early with probability $\po$),}
        \utility_{q_1,\langle1,2\rangle}           & =    \frac{\qo}{2} \cdot \po \cdot (-1-\cost)  \\
        \intertext{(because the uninformed voter votes for $\againstVote$ with probability $\frac{\qo}{2}$ and the informed $\againstVote$-type voter votes early with probability $\po$),}
        \utility_{q_2,\langle1,0\rangle}           & = \qt \cdot (1-\po) \cdot (1-\cost)            \\
        \intertext{(because the informed $\againstVote$-type voter does not vote early with probability $(1-\po)$ and, having seen the $\forVote$ type's early vote, the uninformed voter always votes for $\forVote$ with probability $\qt$),}
        \utility_{q_2,\langle1,1\rangle}           & =
        \frac{\qt}{2} \cdot \po \cdot (1-\cost) +
        \frac{\qt}{2} \cdot \po \cdot (-1-\cost) \nonumber                                          \\&=
        \qt \cdot \po \cdot (-\cost)
        \\
        \intertext{(because the informed $\againstVote$-type voter votes early with probability $\po$ and the uninformed voter decides the vote by voting for $\forVote$ or $\againstVote$ with probability
            $\frac{1}{2}$ each),}
        \utility_{q_\varnothing,\langle1,0\rangle}
                                                   & =
        (1-\qo-\qt) \cdot (1-\po) \cdot (
        (1-\pta) \cdot (1-\cost)
        +
        \pta \cdot (-\cost)
        )
        \nonumber
        \\
                                                   & =
        (1-\qo-\qt) \cdot (1-\po) \cdot (1-\pta - \cost)
        \intertext{(because the uninformed voter abstains with probability $(1-\qo-\qt)$, the informed $\againstVote$-type voter does not vote early with probability $(1-\po)$, and the informed $\forVote$-type voter gets a utility of $(1-\cost)$ or $(-\cost)$ depending on whether the informed $\againstVote$ voter abstains $(1-\pta)$ or votes late $(\pta)$),}
        \utility_{q_\varnothing,\langle1,1\rangle} & = (1-\qo-\qt) \cdot \po \cdot (-\cost)         \\
        \intertext{(because the uninformed voter abstains from voting with probability $(1-\qo-\qt)$ and the informed $\againstVote$-type voter votes early with probability $\po$) \nonumber.}
    \end{align}

    In total, the utility gained from voting in the first turn is:
    \begin{equation*}
        \Uearly =
        (1 - \po) \cdot (1 - \qo + (1 - \qo - \qt)\cdot\pta) - \cost
    \end{equation*}
\end{proof}

\propLateUtility*
\begin{proof}
    If the informed $\forVote$ type voter strategically waits for the next turn, the following interim tallies are of interest:
    \tallyTuple{0}{0}, \tallyTuple{1}{0}, \tallyTuple{0}{1}, \tallyTuple{1}{1}, \tallyTuple{0}{2}.
    The utility $\Unotearly$ can be written as:
    \begin{align}
        \Unotearly                                  & =
        \utility_{q_1, \tallyTuple{1}{0}} +
        \utility_{q_1, \tallyTuple{0}{1}} +
        \utility_{q_1, \tallyTuple{1}{1}} +
        \utility_{q_1, \tallyTuple{0}{2}} +
        \utility_{q_2, \tallyTuple{0}{0}} +
        \utility_{q_2, \tallyTuple{0}{1}} +
        \utility_{q_\varnothing, \tallyTuple{0}{0}} +
        \utility_{q_\varnothing, \tallyTuple{0}{1}},                                                                                                                             \\
        \intertext{where:}
        \utility_{q_1, \tallyTuple{1}{0}}           & = \frac{\qo}{2} \cdot (1-\po) \cdot \left( \ptf \cdot (1-\cost) +  (1-\ptf) \cdot ((1-\pta) \cdot 1 + \pta \cdot 0)\right)
        \nonumber
        \\
                                                    & = \frac{\qo}{2} \cdot (1-\po) \cdot \left( \ptf \cdot (1-\cost) +  (1-\ptf) \cdot (1-\pta)\right)
        \nonumber
        \\
                                                    & =
        \frac{\qo}{2} \cdot (1-\po) \cdot \left( 1-\cost\ptf - \pta (1-\ptf)\right)
        \\ \intertext{(because the uninformed voter turns out to vote in the first turn with probability $\qo$, votes $\forVote$ with probability $\frac{1}{2}$, and furthermore the informed $\againstVote$ voter does not vote early with probability $(1-\po)$; now if the informed $\forVote$ voter votes late with probability $\ptf$, the utility is $1-\cost$ regardless of the informed $\againstVote$ voter's action; otherwise the informed $\forVote$ voter ties or wins depending on whether the informed $\againstVote$ voter votes late, which happens with probability $\pta$ by symmetry),}
        \utility_{q_1, \tallyTuple{0}{1}}
                                                    & =
        \frac{\qo}{2} \cdot (1-\po) \cdot \left(
        (1-\pta) \cdot (-1) +
        \pta \cdot ((1-\ptf) \cdot (-\cost) + \ptf \cdot (-1-\cost))
        \right)
        \nonumber
        \\
                                                    & =
        - \frac{\qo}{2} \cdot (1-\po) \cdot \left(
        1 + \pta \cdot (
            \cost + \ptf - 1
            )
        \right)
        \\ \intertext{(because the uninformed voter votes $\againstVote$ with probability $\frac{\qo}{2}$ and the informed $\againstVote$ voter does not vote early with probability $(1-\po)$; now if the informed $\forVote$ voter abstains, with probability $(1-\pta)$, then the utility is $-1$ regardless of the informed $\againstVote$ voter's action; otherwise the informed $\forVote$ voter loses or ties depending on whether the informed $\againstVote$ voter votes late, which happens with probability $\ptf$ by symmetry, as above),}
        \utility_{q_1, \tallyTuple{1}{1}}           & = \frac{\qo}{2} \cdot \po \cdot (1-\cost)
        \\ \intertext{(because the uninformed voter turns out early with probability $\qo$, votes $\forVote$ with probability $\frac{1}{2}$, the informed $\againstVote$ voter votes early with probability $\po$, and the informed $\forVote$ voter has a dominant strategy to vote late in order to win),}
        \utility_{q_1, \tallyTuple{0}{2}}           & = \frac{\qo}{2} \cdot \po \cdot (-1)
        \\ \intertext{(because the uninformed voter votes $\againstVote$ with probability $\frac{\qo}{2}$, the informed $\againstVote$ voter votes early with probability $\po$, and the informed $\forVote$ voter has a dominant strategy to abstain as they lose regardless of their action),}
        \utility_{q_2, \tallyTuple{0}{0}}           & = \qt \cdot (1-\po) \cdot (
        \ptp \cdot ((1 - \ptp)\cdot (\frac{1}{2}-c)  + \ptp \cdot (-c))  \nonumber                                                                                               \\&
        + (1 - \ptp) \cdot ((1-\ptp) \cdot 0 + \ptp \cdot (-\frac{1}{2})))\nonumber                                                                                              \\&
        = q_2 \cdot (1-\po) \cdot \ptp \cdot (-c)
        \\ \intertext{(because the uninformed voter waits with probability $q_2$ and the informed $\forVote$ voter does not vote early with probability $(1-\po)$; now if the informed $\forVote$ voter votes late with probability $\ptp$ and the informed $\againstVote$ voter abstains with probability $(1-\ptp)$, the utility is either $(1-c)$ or $(-c)$ depending on a coin toss for the uninformed voter; if the informed $\againstVote$ voter votes with probability $\ptp$, the utility is either $(1-c)$ or $(-1-c)$ depending on a coin toss for the uninformed voter; alternatively, if the informed $\againstVote$ voter abstains from voting with probability $(1-\ptp)$, the expected utility is 0 if the informed $\againstVote$ voter abstains as well, or $-\frac{1}{2}$ otherwise),}
        \utility_{q_2, \tallyTuple{0}{1}}
                                                    & =
        \qt \cdot \po \cdot
        \left(
        \pta \cdot (-1-\cost) + (1-\pta) \cdot (-1)
        \right)
        \nonumber
        \\
                                                    & =
        - \qt \cdot \po \cdot
        \left( \cost \pta + 1
        \right)
        \\ \intertext{(because the uninformed voter waits with probability $\qt$ and the informed $\againstVote$ voter votes with probability $\po$; the informed $\forVote$ voter will always lose, but still votes with probability $\pta$ in order to improve their utility in other states of the same information set),}
        \utility_{q_\varnothing, \tallyTuple{0}{0}} & = (1-\qo - \qt) \cdot (1-\po) \cdot (
        \ptp \cdot ((1 - \ptp) \cdot (1-c) + \ptp \cdot (-c)) \nonumber                                                                                                          \\&
        + (1-\ptp) \cdot ((1-\ptp)\cdot 0 + \ptp \cdot (-1))) \nonumber                                                                                                          \\&
        = (1-\qo - \qt) \cdot (1 - \po) \cdot \ptp \cdot (-c)
        \\ \intertext{(because the uninformed voter abstains with probability $(1-\qo-\qt)$ and the informed $\forVote$ voter does not vote early with probability $(1-\po)$; now if the informed $\forVote$ voter votes late with probability $\ptp$, then the utility is either $(1-c)$ or $(-c)$ depending on the informed $\againstVote$ voter's move; alternatively, if the informed $\forVote$ voter abstains from voting with probability $(1-\ptp)$, then the utility is either 0 or $-1$ depending on the informed $\againstVote$ voter's move),}
        \utility_{q_\varnothing, \tallyTuple{0}{1}}
                                                    & =
        (1-\qo-\qt) \cdot \po \cdot
        \left(
        \pta \cdot (-\cost) + (1-\pta) \cdot (-1)
        \right)
        \nonumber
        \\
                                                    & =
        (1-\qo-\qt) \cdot \po \cdot
        \left(
        \pta \cdot (1-\cost) -1
        \right)
        \\ \intertext{(because the uninformed voter abstains with probability $(1-\qo-\qt)$ and the informed $\againstVote$ voter votes early with probability $\po$; the informed $\forVote$ voter ties if they vote with probability $\pta$ and loses otherwise; here again, they do not have a dominant strategy as abstaining may improve utility in other states of the same information set).} \nonumber
    \end{align}
    In total, the utility from waiting in the first turn is:
    \begin{align*}
        \Unotearly
         & =
        (\po-1) \cost \ptp
        -
        \po (
        1 + (\cost-1+\qt) \pta
        )
        \\
         & -
        \frac{\qo}{2} (
        (\cost-2) \po (1 - \pta)
        -
        \cost (
        (\po - 1) (
        2\ptp
        -
        \ptf
        -
        \pta
        )
        -
        \po\pta
        )
        )
    \end{align*}
\end{proof}

\corEarlyWaitDiff*
\begin{proof}
    The stated result is obtained by subtracting \cref{prop:LateUtility} from \cref{prop:EarlyUtility} and simplifying the resulting term:
    \begin{align*}
        \Uearly - \Unotearly
         &
        =
        (1 - \po) \cdot (1 - \qo + (1 - \qo - \qt)\cdot\pta) - \cost
        \\&
        -
        (
        (\po-1) \cost \ptp - \po (1 + (\cost-1+\qt) \pta)
        \\&
        - \frac{\qo}{2} ((\cost-2) \po (1 - \pta) - \cost ((\po - 1) (2\ptp - \ptf - \pta) - \po\pta))
        )
        \\&
        =
        \left(
        \ptp - 1 +
        \left(\frac{1}{2}  \left(\pta + \ptf\right) - \ptp\right) \qo
        \right) \cost
        \nonumber \\&
        - \left(
        \left(3 \pta + \ptf - 2 \ptp - 1\right) \frac{\qo}{2}
        + \left(\ptp - \pta\right)
        \right) \cost \po
        \nonumber \\&
        + 1 - \qo
        - \left(1 - \qo - \qt\right) \pta
    \end{align*}
\end{proof}

\corBothIndifferencePo*
\begin{proof}
    First, we note that the indifference condition for this case is $\Uearly - \Unotearly = 0$.
    We solve \cref{cor:EarlyWaitDiff} under this constraint for $\po$, reaching the stated result.
\end{proof}

\thmQoMixed*
\begin{proof}
    Recall that we have shown in \cref{prop:GreenSet} that voters are indifferent to the value of $\pta$ when $\cost = 1 + \frac{(\ptf \qo - 2 \qt) \po - \ptf \qo}{(2 - 3 \qo)\po + \qo}$.
    We solve this for $\ptf$, and obtain:
    \begin{equation*}
        \ptf = \frac{
            2\po\qt - \left(
            \left(\cost - 1\right)
            \left(\left(3\qo - 2\right)\po - \qo\right)
            \right)
        }{
            \left(\po - 1\right) \qo
        }
    \end{equation*}
    Furthermore, in \cref{prop:WhiteSet} we prove that when $\qo > 0$ and $\cost = \pta$, then voters are indifferent \gls{wrt} $\ptf$.
    We substitute these in \cref{cor:BothIndifferencePo}, and note that these two substitutions can co-exist as \cref{prop:GreenSet} specifies that in our case, voters are indifferent to $\pta$.
    After some arithmetic, we arrive at the following form for $\po$:
    \begin{align*}
        \po =
        1 + \frac{2 + \left(\cost - 2\right)\qo -2\cost}{2\cost\left(\left(1-\qo\right)\ptp + \qo + \qt - 1\right)}
    \end{align*}
    Note that the above form is only dependent on $\ptp$, and our exogenous parameters $\cost,\qo,\qt$.
    To obtain our result, we need to ensure that our expressions for both $\ptf$ and $\po$ are feasible.
    Thus, we require that $\ptf \in \left[0,1\right]$ and that $\po \in \left(0,1\right)$ (as we are seeking fully mixed equilibria):
    \begin{align*}
        \frac{
            2\po\qt - \left(
            \left(\cost - 1\right)
            \left(\left(3\qo - 2\right)\po - \qo\right)
            \right)
        }{
            \left(\po - 1\right) \qo
        }
         &
        \in
        \left[0, 1\right]
        \\
        1 + \frac{2 + \left(\cost - 2\right)\qo -2\cost}{2\cost\left(\left(1-\qo\right)\ptp + \qo + \qt - 1\right)}
         &
        \in
        \left(0, 1\right)
    \end{align*}
    The requirement for $\ptf$ results in:
    \begin{align*}
        \cost \in \left[
            2\po\frac{\qo+\qt-1}{\left(3\qo-2\right)\po -\qo},
            1 + \frac{2\qt\po}{\left(3\qo-2\right)\po -\qo}
            \right]
    \end{align*}
    And, with respect to $\po$, we have:
    \begin{align*}
        \ptp
         &
        \in
        \left[0,1+\frac{\qt}{\qo-1}\right)
            \cup
            \left(1+\frac{\qt}{\qo-1},1\right]
        \\
        \cost
         &
        \in
        \left(
        \frac{2-2\qo}{2+2\left(\qo-1\right)\ptp -3\qo -2\qt}
        ,
        1+\frac{\qo}{\qo-2}
        \right)
    \end{align*}
    By intersecting these parameter ranges, it becomes apparent that the intersection is non-empty for certain parameter regimes.
    In particular, one possible parameter set is $\cost = \frac{1}{2}$, $\qo = \frac{43}{64}$, $\qt = \frac{169}{768}$, with the corresponding strategy being $\pta = \frac{1}{2}$, $\ptp = \frac{107}{252}$, $\ptf = 0$, $\po = \frac{3}{4}$.
\end{proof}

\thmQoZeroStrategiesTurnOne*
\begin{proof}
    We start with \cref{prop:QoZeroStrategiesTurnTwo}, which reasons about second-turn strategies when $\qo = 0$.
    \propQoZeroStrategiesTurnTwo*
    The proof of \cref{prop:QoZeroStrategiesTurnTwo} is given in \cref{sec:Proofs}.
    We continue by substituting $\qo = 0$ in \cref{cor:EarlyWaitDiff} to obtain the difference in utility between voting early and late in the current case, and get:
    \begin{align}
        \label{eq:EarlyWaitDiffQoZero}
        \Uearly - \Unotearly
         &
        =
        \left(
        \ptp - 1 +
        \left(\frac{1}{2}  \left(\pta + \ptf\right) - \ptp\right) 0
        \right) \cost
        \nonumber \\&
        - \left(
        \left(3 \pta + \ptf - 2 \ptp - 1\right) \frac{0}{2}
        + \left(\ptp - \pta\right)
        \right) \cost \po
        \nonumber \\&
        + 1 - 0
        - \left(1 - 0 - \qt\right) \pta
        \nonumber \\&
        =
        \left( \ptp - 1 \right) \cost - \left(\ptp - \pta\right)\cost \po + 1 - \left(1 - \qt\right) \pta
    \end{align}
    Note that \cref{eq:EarlyWaitDiffQoZero} implies that voting early dominates waiting when $\pta=\ptp$:
    \begin{itemize}
        \item When $\pta=1=\ptp$, the difference in utility is $\qt$, implying that later uninformed arrivals contribute to the informed voters' decision to vote early.
        \item When $\pta=0=\ptp$, the difference in utility equals $1-\cost$, which is strictly positive as $\cost \in \left[0,1\right)$.
              Furthermore, the utility is obviously monotonically increasing in $\cost$.
    \end{itemize}
    We now attempt to analyze cases where voters are indifferent between voting early and late.
    Thus, we substitute $\qo = 0$ in \cref{cor:BothIndifferencePo}, :
    \begin{align}
        \label{eq:QoZeroStrategiesTurnOne}
        \po
        =
        \frac{
            (1-\ptp)\cost + (1-\qt)\pta - 1
        }{
            (\pta-\ptp)\cost
        }
    \end{align}
    Given $\qt$ and $\cost$, we can use \cref{eq:QoZeroStrategiesTurnTwo} to derive equilibria values for $\pta$ and $\ptp$.
    These, in turn, can be substituted in \cref{eq:QoZeroStrategiesTurnOne} to obtain equilibria values for $\po$.
    Having ruled out the cases where $\ptp = \pta$, we turn to substitute $\ptp = 1$ and $\pta = 0$ in \cref{eq:QoZeroStrategiesTurnOne}, and receive $\po = \frac{1}{\cost}$.
    But, \cref{prop:QoZeroStrategiesTurnTwo} shows that these values are only for $\cost \in \left(1-\qt,1-\frac{\qt}{2}\right)$, implying that $\frac{1}{\cost} > 1$, which is not feasible.
    Note that substituting $\ptp=1$ and $\pta=0$ into the utility difference formula results in a utility of $1$, implying that voting dominates.
    The proof of \cref{prop:QoZeroStrategiesTurnTwo} gives us two additional cases to examine.
    The first is when $\cost = 1 - \qt$, $\ptp = 1$ and voters are indifferent \gls{wrt} $\pta$, allowing us to get:
    $\po = \frac{\pta(\qt - 1) + 1}{(\pta - 1)(\qt - 1)}$.
    By requiring that $\qo \le 1$ and solving the resulting equation: $\frac{\pta(\qt - 1) + 1}{(\pta - 1)(\qt - 1)} \le 1$, one can see that the condition is satisfied only for $\qt = 0$, and even then we get that $\qo = 1$.
    The second and final case is when voters are indifferent \gls{wrt} $\ptp$, $\pta = 0$ and $\cost = 1 - \frac{\qt}{2}$, which imply $\po = 1 - \frac{\qt}{\ptp(\qt - 2)}$, but again we reach a degenerate case that is only feasible when $\qt = 0$, and which implies $\po = 1$.
\end{proof}

\propQoZeroStrategiesTurnTwo*
\begin{proof}
    \cref{prop:MustardSet,prop:GreenSet,prop:WhiteSet} show that equilibrium strategies may vary in different regions of our parameter space (i.e., $\cost$ and $\vec{\uninformed}$).
    In particular, \cref{prop:MustardSet} shows that the values assigned in equilibria to $\ptp$ depend solely on our exogenous parameters for the entire parameter space.
    On the other hand, \cref{prop:WhiteSet} shows that this is true only when $\qo = 0$ for $\ptf$, in which case voters are indifferent to the value of $\ptf$.
    This is because when $\qo = 0$, uninformed voters do not vote early, implying that the information set that $\ptf$ applies to does not exist.
    This extends to $\pta$ and $\ptp$, as well; by substituting $\qo = 0$ in \cref{prop:GreenSet,prop:MustardSet}, we reach:
    \begin{equation}
        \label{eq:QoZeroStrategiesTurnTwo}
        \pta
        =
        \begin{cases}
            1,                   & \cost < 1 - \qt
            \\
            0,                   & \cost > 1 - \qt
            \\
            \text{indifference}, & \cost = 1 - \qt
        \end{cases}
        ,
        \qquad
        \ptp
        =
        \begin{cases}
            1,                   & \cost < 1 - \frac{\qt}{2}
            \\
            0,                   & \cost > 1 - \frac{\qt}{2}
            \\
            \text{indifference}, & \cost = 1 - \frac{\qt}{2}
        \end{cases}
    \end{equation}
    We turn to characterize the feasible combinations for $\ptp$ and $\pta$ by using \cref{eq:QoZeroStrategiesTurnTwo}.
    \begin{enumerate}
        \item If voters are indifferent \gls{wrt} both $\pta$ and $\ptp$, then $\cost = 1-\qt = 1 - \frac{\qt}{2}$, but that is possible only when $\qt = 0$.
              Moreover, this implies that $\cost = 1$, which is outside our feasible parameter range, thereby ruling-out this possibility.
        \item If voters are both indifferent \gls{wrt} $\pta$ and $\ptp = 0$, then this implies $\cost = 1 - \qt > 1 - \frac{\qt}{2}$, which in turn implies $\qt < 0$, which cannot be as $\qt$ is a probability.
        \item If voters are indifferent \gls{wrt} $\pta$ and $\ptp = 1$, then by combining the constraints for both, we get $\cost = 1-\qt < 1 -\frac{\qt}{2}$, implying that this can only hold when $0 < \qt$.
        \item If $\pta = 0$ and voters are indifferent \gls{wrt} $\ptp$, then $1-\frac{\qt}{2} = \cost > 1 - \qt$, therefore the current condition can only be satisfied when $\qt > 0$.
        \item If $\pta = 0$ and $\ptp = 0$, we can deduce $\cost > \max\left(1 - \qt,1-\frac{\qt}{2}\right)$, which in our case can be written succinctly: $\cost > 1 - \frac{\qt}{2}$.
        \item If $\pta = 0$ and $\ptp = 1$, then we get a feasible range: $ 1-\frac{\qt}{2} > \cost > 1 - \qt$, where this range is non-empty for $\qt > 0$.
        \item If $\pta = 1$ and voters are indifferent \gls{wrt} $\ptp$, we get a contradiction: $1-\frac{\qt}{2} = \cost < 1 - \qt$.
        \item If $\pta = 1$ and $\ptp = 0$, a contradiction is reached: $1-\frac{\qt}{2} < \cost < 1 - \qt$.
        \item If $\pta = 1$ and $\ptp = 1$, we can deduce: $\cost < \min\left(1 - \qt, 1-\frac{\qt}{2}\right)$, which in the current case can be abbreviated to: $\cost < 1 - \qt$.
              For $\cost$ to be feasible, we must also require that $\qt < 1$.
    \end{enumerate}

\end{proof}

\thmWaitingEquilibria*
\begin{proof}
    We prove this result by ruling out almost all strategies for the current case in \cref{prop:WaitingIndifference}.
    Then, we prove that there are feasible parameter regimes which admit such equilibria, and characterize the equilibria strategies for each region of our parameter space.

    Notably, \cref{prop:WaitingIndifference} implies that in the current case, voters exhibit a ``free-rider'' mentality that is dependent on the cost of voting: if $\cost$ is high, then voters assign a low probability of voting, even when the interim tally shows their preference is the leading candidate.
    \begin{restatable}[]{proposition}{propWaitingIndifference}
        \label{prop:WaitingIndifference}
        In an equilibrium strategy profile where waiting dominates voting early, if $\qo > 0$, then $\po = 0$, $\pta = \cost$ and $\ptf = 1 - \cost$.
    \end{restatable}
    \begin{proof}
        Given the current assumption that waiting dominates early voting, we analyze all possible strategies for the case where $\qo > 0$, as given by \cref{eq:P2aWait,prop:MustardSet,prop:WhiteSet}.
        We show that the conditions implied by for all but one strategy allow us to reach contradictions.

        Our assumption that waiting dominates early voting implies $\po = 0$, and that $\Uearly - \Unotearly < 0$.
        By substituting both into \cref{cor:EarlyWaitDiff}, we can constrain the feasible voting strategies:
        \begin{equation}
            \label{eq:WaitConstraint}
            (
            (\pta + \ptf - 2 \cdot \ptp) \frac{\cost}{2}
            + \pta - 1
            ) \qo
            + (\ptp-1)\cost + (\qt - 1) \pta + 1
            < 0
        \end{equation}
        Furthermore, substituting $\po = 0$ into \cref{prop:GreenSet} allows to simplify $\pta$:
        \begin{equation}
            \label{eq:P2aWait}
            \pta
            =
            \begin{cases}
                1,                   & \cost < 1 - \ptf
                \\
                0,                   & \cost > 1 - \ptf
                \\
                \text{indifference}, & \cost = 1 - \ptf
            \end{cases}
        \end{equation}
        Armed with \cref{eq:WaitConstraint,eq:P2aWait}, we tackle each possible strategy.
        \begin{enumerate}[leftmargin=*]
            \item If $\pta = 1$ and $\ptf = 1$ dominate, then by \cref{prop:WhiteSet} we get that $\cost < \pta = 1$, and similarly from \cref{eq:P2aWait} we obtain that $\cost < 1 - \ptf = 0$.
                  Both imply that setting $\pta = 1$ and $\ptf = 1$ is only rational when $\cost < 0$, but this is outside our feasible parameter range.

            \item If $\pta = 1$ and $\ptf = 0$ dominate, then \cref{prop:WhiteSet,eq:P2aWait} correspondingly tell us that $c > \pta = 1$ and $1 = 1 - \ptf > c$, implying a contradiction: $1 > c > 1$.

            \item If $\pta = 1$ dominates and voters are indifferent \gls{wrt} $\ptf$, then as before,  we use \cref{prop:WhiteSet,eq:P2aWait} to reach a contradiction: $1 = \pta = \cost < 1 - \ptf \le 1$.

            \item If $\pta = 0$ and $\ptf = 1$ dominate, then \cref{prop:WhiteSet,eq:P2aWait} allow us to arrive at a contradiction: $0 = 1 - \ptf < c < \pta = 0$.

            \item If $\pta = 0$ and $\ptf = 0$ dominate, then from \cref{eq:P2aWait}, we deduce that $\cost > 1 - \ptf = 1$, which is outside our feasible region for costs.

            \item If $\pta = 0$ dominates and voters are indifferent \gls{wrt} $\ptf$, then \cref{prop:WhiteSet,eq:P2aWait} imply: $0 = \pta = \cost < 1 - \ptf$.
                  Substituting this in \cref{eq:WaitConstraint} gives us:
                  \begin{align*}
                      0
                       &
                      >
                      (
                      (\pta + \ptf - 2 \cdot \ptp) \frac{\cost}{2}
                      + \pta - 1
                      ) \qo
                      + (\ptp-1)\cost + (\qt - 1) \pta + 1
                      \nonumber \\&
                      =
                      (
                      (\pta + \ptf - 2 \cdot \ptp)\cdot0
                      + 0 - 1
                      ) \qo
                      + (\ptp-1)\cdot0 + (\qt - 1)\cdot0 + 1
                      \nonumber \\&
                      =
                      1 - \qo
                  \end{align*}
                  But, because $\qo \in \left[0,1\right]$, we reach a contradiction:
                  $0 > 1 - \qo \geq 0$.

            \item If voters are indifferent \gls{wrt} $\pta$ and $\ptf = 1$ dominates, then \cref{prop:WhiteSet,eq:P2aWait} and the fact that $\pta \in \left[0,1\right]$ produce a contradiction: $\pta < \cost = 1 - \ptf = 0$.

            \item If voters are indifferent \gls{wrt} $\pta$ and $\ptf = 0$ dominates, then \cref{eq:P2aWait} implies $\cost = 1 - \ptf = 1$, which is outside the feasible regime for $\cost$.

            \item If voters are indifferent \gls{wrt} to both $\pta$ and $\ptf$, then we get from \cref{prop:WhiteSet,eq:P2aWait} that $\pta = \cost = 1 - \ptf$.
                  Thus, we deduce that $\pta = \cost$ and $\ptf = 1 - \cost$.
        \end{enumerate}
        In total, the only feasible strategy is the one in which $\pta = \cost$ and $\ptf = 1 - \cost$.
    \end{proof}

    Now, we substitute the strategy given by \cref{prop:WaitingIndifference} in \cref{eq:WaitConstraint}:
    \begin{equation}
        \label{eq:EarlyDominatesFinalIndifference}
        \left(\ptp\left(1 - \qo\right) - 2 + 1.5\qo + \qt\right)\cost + 1 - \qo < 0
    \end{equation}
    Given this new constraint, we analyze each option given by \cref{prop:MustardSet} for $\ptp$, and characterize the feasible parameter regimes that admit each via a series of results.
    The proofs for the following results are given in \cref{sec:Proofs}.

    \propWaitingPtpZero*

    \propWaitingPtpOne*
\end{proof}

\propWaitingPtpZero*
\begin{proof}
    According to \cref{prop:MustardSet}, our choice of $\ptp$ imposes two conditions: $\qo < 1$ and $\cost > 1 + \frac{\qt}{2 \left(\qo - 1\right)} = \frac{2\qo + \qt - 2}{2 \left(\qo - 1\right)}$.
    We substitute $\ptp = 0$ in \cref{eq:EarlyDominatesFinalIndifference}, and get:
    $\cost\left(1.5\qo + \qt - 2\right) < \qo - 1$.
    This implies that: $\cost > \frac{\qo - 1}{1.5\qo + \qt - 2}$, where the direction of the inequality was flipped because both $\qo + \qt \le 1$ and $\qo < 1$, allowing us to deduce: $1.5\qo + \qt - 2 \le -\frac{1}{2}$.
    Together with the condition imposed by \cref{prop:MustardSet}, we get:
    $\cost \in \left(
        \max\left(
            \frac{\qo - 1}{1.5\qo + \qt - 2}
            ,
            \frac{2\qo + \qt - 2}{2 \left(\qo - 1\right)}
            \right)
        ,
        1
        \right)$.

    This range is non-empty when either $\frac{\qo - 1}{1.5\qo + \qt - 2} \le \frac{2\qo + \qt - 2}{2 \left(\qo - 1\right)} < 1$, or $\frac{2\qo + \qt - 2}{2 \left(\qo - 1\right)} \le \frac{\qo - 1}{1.5\qo + \qt - 2} < 1$.

    First, examine the case where: $\frac{\qo - 1}{1.5\qo + \qt - 2} < 1$.
    This requirement is equivalent to $\qo -1 > 1.5\qo + \qt - 2$, where the direction of the inequality is flipped as $1.5\qo + \qt - 2 \le -\frac{1}{2} < 0$.
    This requirement implies that $1 > 0.5\qo + \qt$, which is indeed true: by definition $\qo + \qt \le 1 $, and in the current case we furthermore have $\qo > 0$.

    In the other case, we have the following condition: $\frac{2\qo + \qt - 2}{2 \left(\qo - 1\right)} < 1$.
    When multiplying both sides of by $\qo -1$, we flip the direction of the inequality due to our assumption that $\qo \in \left( 0, 1 \right)$, giving us: $2\qo + \qt - 2 > 2 \left(\qo - 1\right)$.
    This implies that: $\qt > 0$.

    Due to the asymmetry between the feasible $\qt$ values for the two cases, we characterize in \cref{prop:WaitingPtpZeroCostConstraintFirst,prop:WaitingPtpZeroCostConstraintSecond} the regimes for which each one is satisfied.
    Note that these results are slightly more general as they are reused for additional proofs.
    In particular, \cref{prop:WaitingPtpZeroCostConstraintFirst} does not impose the restriction $\qo > 0$, which is instead accounted for only afterward.

    \begin{restatable}[]{proposition}{propWaitingPtpZeroCostConstraintFirst}
        \label{prop:WaitingPtpZeroCostConstraintFirst}
        We have that $\frac{\qo - 1}{1.5\qo + \qt - 2} < \frac{2\qo + \qt - 2}{2 \left(\qo - 1\right)}$ is satisfied when $\qt \in \left[0, 2 - \sqrt{2} \right)$, and $\qo \in \left[0, \frac{3-3.5\qt - \sqrt{8.25\qt^2 -5\qt + 1}}{2}\right) \cup \left(\frac{3-3.5\qt + \sqrt{8.25\qt^2 -5\qt + 1}}{2}, 1\right)$.
    \end{restatable}

    \begin{restatable}[]{proposition}{propWaitingPtpZeroCostConstraintSecond}
        \label{prop:WaitingPtpZeroCostConstraintSecond}
        We have that $\frac{\qo - 1}{1.5\qo + \qt - 2} \geq \frac{2\qo + \qt - 2}{2 \left(\qo - 1\right)}$ is satisfied when $\qt \in \left[0,1\right]$, and
        $\qo \in \left[\frac{3-3.5\qt - \sqrt{8.25\qt^2 -5\qt + 1}}{2}, \frac{3-3.5\qt + \sqrt{8.25\qt^2 -5\qt + 1}}{2} \right]$.
    \end{restatable}

    Proofs for both results are given in \cref{sec:Proofs}.
    Intuitively, we use the various requirements to constrain our parameter space, resulting in systems of several equations which we solve.
    By using them, we fully characterize when $\frac{\qo - 1}{1.5\qo + \qt - 2} < \frac{2\qo + \qt - 2}{2 \left(\qo - 1\right)}$ holds (via \cref{prop:WaitingPtpZeroCostConstraintFirst}), and similarly, where we have $\frac{\qo - 1}{1.5\qo + \qt - 2} \geq \frac{2\qo + \qt - 2}{2 \left(\qo - 1\right)}$ (via \cref{prop:WaitingPtpZeroCostConstraintSecond}).

    Next, we account for the constraints which were left out of the aforementioned two results, to ensure their generality.
    In particular, we have to consider that $\qo + \qt \in \left[0,1\right]$, and also consider the requirement of the current case that $\qo \in \left( 0, 1 \right)$.
    For the former constraint, one can verify using the extreme values of our range for $\qt$ and the resulting ranges for $\qo$ that the range of values $\qo \in \left(\frac{3-3.5\qt + \sqrt{8.25\qt^2 -5\qt + 1}}{2}, 1\right)$ is infeasible.
    Furthermore, for the latter constraint, one can account for it by intersecting it with our resulting parameter regions.
    Finally, the characterization stated in \cref{prop:WaitingPtpZero} is obtained.
\end{proof}

\propWaitingPtpZeroCostConstraintFirst*
\begin{proof}
    For the given range of $\qo$ values, we have that $\qo \in \left[0,1\right)$.
    Thus, we have $1.5\qo + \qt - 2 < 0$ and $\qo - 1 < 0$, implying that $\frac{\qo - 1}{1.5\qo + \qt - 2} < \frac{2\qo + \qt - 2}{2 \left(\qo - 1\right)}$ is equivalent to:
    \begin{align*}
        0
         &
        <
        \left(2\qo + \qt - 2\right)\left(1.5\qo + \qt - 2\right)
        -
        2 \left(\qo - 1\right)^2
        \nonumber \\&
        =
        \qo^2 + \left(3.5\qt - 3\right)\qo + \qt^2 - 4\qt + 2
    \end{align*}
    The roots of this polynomial are obtained at:
    \begin{align*}
        \qo
         &
        =
        \frac{
            3-3.5\qt \pm \sqrt{
                (3.5\qt - 3)^2 - 4(\qt^2 - 4\qt + 2)
            }
        }{2}
        \nonumber \\&
        =
        \frac{
            3-3.5\qt \pm \sqrt{
                12.25\qt^2 - 21\qt + 9 - 4\qt^2 + 16\qt - 8
            }
        }{2}
        \nonumber \\&
        =
        \frac{1}{2}\left(3-3.5\qt \pm \sqrt{8.25\qt^2 - 5\qt + 1}\right)
    \end{align*}
    We always have two real roots for $\qo$, as:
    \begin{align}
        \label{eq:EarlyDominatesFinalIndifferenceConditionSqrtIsPositive}
        8.25\qt^2 -5\qt + 1
        =
        \left(\sqrt{8.25}\qt - 1\right)^2 + \left(\sqrt{33}-5\right)\qt
        >
        \left(\sqrt{25}-5\right)\qt
        \geq
        0
    \end{align}
    Thus far, our feasible range for $\qo$ is:
    \begin{equation}
        \label{eq:EarlyDominatesQo}
        \qo \in
        \left(
        0
        ,
        \frac{1}{2}\left(
            3-3.5\qt
            -
            \sqrt{8.25\qt^2 -5\qt + 1}
            \right)
        \right)
        \bigcup
        \left(
        \frac{1}{2}\left(
            3-3.5\qt
            +
            \sqrt{8.25\qt^2 -5\qt + 1}
            \right)
        ,
        1
        \right)
    \end{equation}
    The left segment of \cref{eq:EarlyDominatesQo} is non-empty when
    $0 < \frac{1}{2}\left(
        3-3.5\qt
        -
        \sqrt{8.25\qt^2 -5\qt + 1}
        \right)$, or equivalently:
    $
        \frac{1}{2}\sqrt{8.25\qt^2 -5\qt + 1}
        <
        \frac{1}{2}\left( 3-3.5\qt \right)
    $.
    Per \cref{eq:EarlyDominatesFinalIndifferenceConditionSqrtIsPositive}, the term $8.25\qt^2 -5\qt + 1$ is strictly positive, while the term $3-3.5\qt$ is non-negative for $\qt \le \frac{3}{3.5}$, thus the inequality does not hold for $\qt > \frac{3}{3.5}$.
    Given that $\qt \leq \frac{3}{3.5}$, we can raise the two sides by a power of two:
    \begin{align*}
        0
         &
        <
        \frac{1}{4} \left(
        \left( 3-3.5\qt \right)^2
        -
        \left( 8.25\qt^2 -5\qt + 1 \right)
        \right)
        \nonumber \\&
        =
        \frac{12.25\qt^2 - 21\qt + 9 - 8.25\qt^2 +5\qt - 1}{4}
        \nonumber \\&
        =
        \qt^2 - 4\qt + 2
    \end{align*}
    This equals zero when:
    $
        \qt
        =
        \frac{1}{2}\cdot\left(4 \pm \sqrt{16-4\cdot2}\right)
        =
        \frac{1}{2}\cdot\left(4 \pm 2\sqrt{2}\right)
        =
        2 \pm \sqrt{2}
    $, which implies that either $\qt < 2-\sqrt{2}$, or $\qt > 2+\sqrt{2}$.
    But, $\qt \le 1 < 2+\sqrt{2}$, so we remain with $\qt < 2-\sqrt{2}$.

    The right segment of \cref{eq:EarlyDominatesQo} is non-empty when:
    $\frac{1}{2}\left(
        3-3.5\qt
        +
        \sqrt{8.25\qt^2 -5\qt + 1}
        \right) < 1$, which can be simplified to:
    $\sqrt{8.25\qt^2 -5\qt + 1} < 3.5\qt - 1$.
    From \cref{eq:EarlyDominatesFinalIndifferenceConditionSqrtIsPositive}, the left-hand side is strictly positive, implying that the inequality does not hold for $\qt \le \frac{1}{3.5} \approx 0.28$.
    When $\qt > \frac{1}{3.5} \approx 0.28$, then both sides are non-negative and we can raise them to the power of 2:
    $8.25\qt^2 -5\qt + 1 < \left(3.5\qt - 1\right)^2$, which can be simplified to:
    \begin{align*}
        0
         &
        <
        12.25\qt^2 - 7\qt + 1 - \left( 8.25\qt^2 -5\qt + 1 \right)
        \nonumber \\&
        =
        4\qt^2-2\qt
        \nonumber \\&
        =
        4\qt\left(\qt - \frac{1}{2}\right)
    \end{align*}
    This implies $\qt >\frac{1}{2}$, meaning that for this case we have $\qt \in \left(\frac{1}{2}, 1\right)$.

    We need at least one segment of \cref{eq:EarlyDominatesQo} to be non-empty, so we suffice with $\qt < 2-\sqrt{2} \approx 0.58$.
    In total, our feasible range of parameters is:
    \begin{align*}
        \qo
         &
        \in
        \left[
        0
        ,
        \frac{1}{2}\left(
        3-3.5\qt
        -
        \sqrt{8.25\qt^2 -5\qt + 1}
        \right)
        \right)
        \bigcup
        \left(
        \frac{1}{2}\left(
            3-3.5\qt
            +
            \sqrt{8.25\qt^2 -5\qt + 1}
            \right)
        ,
        1
        \right)
        \\
        \qt
         &
        \in \left[0, 2 - \sqrt{2} \right)
    \end{align*}
\end{proof}

\propWaitingPtpZeroCostConstraintSecond*
\begin{proof}
    Similar reasoning to the proof of \cref{prop:WaitingPtpZeroCostConstraintFirst} allows us to deduce that $\qo$ values must be in the following range:
    \begin{equation}
        \label{eq:WaitingPtpZeroCostConstraintSecondRangeQo}
        \qo
        \in
        \left[
            \frac{1}{2}\left(
            3-3.5\qt
            -
            \sqrt{8.25\qt^2 -5\qt + 1}
            \right)
            ,
            \frac{1}{2}\left(
            3-3.5\qt
            +
            \sqrt{8.25\qt^2 -5\qt + 1}
            \right)
            \right]
    \end{equation}
    In particular, we get from \cref{eq:EarlyDominatesFinalIndifferenceConditionSqrtIsPositive} that this range is non-empty for $\qt \in \left[0,1\right]$, i.e.
    \begin{equation*}
        \frac{1}{2}\left(
        3-3.5\qt
        -
        \sqrt{8.25\qt^2 -5\qt + 1}
        \right)
        <
        \frac{1}{2}\left(
        3-3.5\qt
        +
        \sqrt{8.25\qt^2 -5\qt + 1}
        \right).
    \end{equation*}
    But, for $\qo$ to be feasible, we also have to require that either one of \cref{eq:WaitingPtpZeroCostConstraintSecondLe1,eq:WaitingPtpZeroCostConstraintSecondGe0,eq:WaitingPtpZeroCostConstraintSecondSurround} must hold.
    \begin{equation}
        \label{eq:WaitingPtpZeroCostConstraintSecondLe1}
        0 \le \frac{1}{2}\left(3-3.5\qt - \sqrt{8.25\qt^2 -5\qt + 1}\right) \le 1
    \end{equation}
    \begin{equation}
        \label{eq:WaitingPtpZeroCostConstraintSecondGe0}
        0 \le \frac{1}{2}\left(3-3.5\qt + \sqrt{8.25\qt^2 -5\qt + 1}\right) \le 1
    \end{equation}
    \begin{equation}
        \label{eq:WaitingPtpZeroCostConstraintSecondSurround}
        \frac{1}{2}\left(3-3.5\qt - \sqrt{8.25\qt^2 -5\qt + 1}\right)
        \le 0 < 1 \le
        \frac{1}{2}\left(3-3.5\qt + \sqrt{8.25\qt^2 -5\qt + 1}\right)
    \end{equation}

    \paragraphNoSkip{Satisfying \cref{eq:WaitingPtpZeroCostConstraintSecondGe0}}
    For \cref{eq:WaitingPtpZeroCostConstraintSecondGe0} to hold, we can equivalently require that both \cref{eq:WaitingPtpZeroCostConstraintSecondGe0Ge0,eq:WaitingPtpZeroCostConstraintSecondGe0Le1} hold:
    \begin{equation}
        \label{eq:WaitingPtpZeroCostConstraintSecondGe0Ge0}
        \sqrt{8.25\qt^2 -5\qt + 1} \geq 3.5\qt - 3
    \end{equation}
    \begin{equation}
        \label{eq:WaitingPtpZeroCostConstraintSecondGe0Le1}
        \sqrt{8.25\qt^2 -5\qt + 1} \le 3.5\qt - 1
    \end{equation}
    For \cref{eq:WaitingPtpZeroCostConstraintSecondGe0Le1}, the proof of \cref{prop:WaitingPtpZeroCostConstraintFirst} shows it is satisfied when $\qt \in \left[\frac{1}{2},1\right]$.
    The condition given in \cref{eq:WaitingPtpZeroCostConstraintSecondGe0Ge0} is satisfied when $\qt \in \left[0,\frac{3}{3.5}\right]$ due to the strict positivity of the left-hand side, as shown by \cref{eq:EarlyDominatesFinalIndifferenceConditionSqrtIsPositive}.
    For $\qt > \frac{3}{3.5}$, this is satisfied when:
    \begin{align}
        \label{eq:WaitingPtpZeroCostConstraintSecondGe0GtFrac}
        0
         &
        \geq
        \left(3.5\qt - 3\right)^2
        -
        \left(8.25\qt^2 -5\qt + 1\right)
        \nonumber \\&
        =
        \left(12.25\qt^2 - 21\qt + 9\right)
        -
        \left(8.25\qt^2 -5\qt + 1\right)
        \nonumber \\&
        =
        4\qt^2 - 16\qt + 8
        \nonumber \\&
        =
        4\cdot\left(\qt^2 - 4\qt + 2\right)
    \end{align}
    From the proof of \cref{prop:WaitingPtpZeroCostConstraintFirst}, we get that \cref{eq:WaitingPtpZeroCostConstraintSecondGe0GtFrac} is satisfied when $\qt \in \left[2-\sqrt{2}, 2+\sqrt{2}\right]$.
    But, $0.58 \approx 2-\sqrt{2} < \frac{3}{3.5} \approx 0.85$, and furthermore we must have $\qt \le 1 < 2+\sqrt{2} \approx 3.41$.
    Therefore, we find that \cref{eq:WaitingPtpZeroCostConstraintSecondGe0Ge0} holds when:
    $\qt \in \left[0,\frac{3}{3.5}\right]\bigcup\left(\frac{3}{3.5}, 1\right] = \left[0,1\right]$.

    Recall that the satisfiability of \cref{eq:WaitingPtpZeroCostConstraintSecondGe0} hinges on both \cref{eq:WaitingPtpZeroCostConstraintSecondGe0Ge0,eq:WaitingPtpZeroCostConstraintSecondGe0Le1} holding.
    Thus, we intersect their feasible parameter ranges.
    In total, we get that \cref{eq:WaitingPtpZeroCostConstraintSecondGe0} holds when $\qt \in \left[\frac{1}{2},1\right]$.

    \paragraphNoSkip{Satisfying \cref{eq:WaitingPtpZeroCostConstraintSecondLe1}}
    For \cref{eq:WaitingPtpZeroCostConstraintSecondLe1} to hold, we can require that both of the following are satisfied:
    \begin{equation}
        \label{eq:WaitingPtpZeroCostConstraintSecondLe1Ge0}
        \sqrt{8.25\qt^2 -5\qt + 1} \le 3 - 3.5\qt
    \end{equation}
    \begin{equation}
        \label{eq:WaitingPtpZeroCostConstraintSecondLe1Le1}
        \sqrt{8.25\qt^2 -5\qt + 1} \geq 1 - 3.5\qt
    \end{equation}
    For \cref{eq:WaitingPtpZeroCostConstraintSecondLe1Ge0}, from the proof of \cref{prop:WaitingPtpZeroCostConstraintFirst} we get that the condition holds for $\qt \le 2-\sqrt{2}$.
    The constraint given in \cref{eq:WaitingPtpZeroCostConstraintSecondLe1Le1} is always satisfied when $\qt \geq \frac{1}{3.5}$ due to \cref{eq:EarlyDominatesFinalIndifferenceConditionSqrtIsPositive}.
    With respect to $\qt < \frac{1}{3.5}$, both sides are positive, allowing us to raise them to the power of two, and solve the resulting inequality:
    \begin{align*}
        0
         &
        \le
        8.25\qt^2 -5\qt + 1 - \left(1 - 3.5\qt\right)^2
        \\&
        =
        8.25\qt^2 -5\qt + 1 - \left(1 - 7\qt + 12.25\qt^2\right)
        \\&
        =
        -4\qt^2 + 2\qt
        \\&
        =
        -4\qt \left(\qt - \frac{1}{2}\right)
    \end{align*}
    This holds for $\qt \in \left[0,\frac{1}{2}\right]$.
    As \cref{eq:WaitingPtpZeroCostConstraintSecondLe1Le1} holds for $\qt \geq \frac{1}{3.5}$ and $\frac{1}{3.5} < \frac{1}{2}$, thus we have $\qt \in \left[0,1\right]$.

    In total, note that \cref{eq:WaitingPtpZeroCostConstraintSecondLe1} is satisfied only when both \cref{eq:WaitingPtpZeroCostConstraintSecondLe1Ge0,eq:WaitingPtpZeroCostConstraintSecondLe1Le1} are satisfied.
    By intersecting their feasible value ranges, we find that \cref{eq:WaitingPtpZeroCostConstraintSecondLe1} holds for $\qt \in \left[0,2-\sqrt{2}\right]$.

    \paragraphNoSkip{Satisfying \cref{eq:WaitingPtpZeroCostConstraintSecondSurround}}
    The condition given in \cref{eq:WaitingPtpZeroCostConstraintSecondSurround} is equivalent to requiring that both:
    \begin{equation}
        \label{eq:WaitingPtpZeroCostConstraintSecondSurroundLe}
        \sqrt{8.25\qt^2 -5\qt + 1} \geq 3-3.5\qt
    \end{equation}
    \begin{equation}
        \label{eq:WaitingPtpZeroCostConstraintSecondSurroundGe}
        \sqrt{8.25\qt^2 -5\qt + 1} \geq 3.5\qt - 1
    \end{equation}
    The condition of \cref{eq:WaitingPtpZeroCostConstraintSecondSurroundLe} holds for $\qt \geq \frac{3}{3.5}$, due to \cref{eq:EarlyDominatesFinalIndifferenceConditionSqrtIsPositive}.
    For $\qt < \frac{3}{3.5}$, reasoning similar to that given in the proof of \cref{prop:WaitingPtpZeroCostConstraintFirst} shows that it is satisfied for $\qt \in \left[2-\sqrt{2}, 2+\sqrt{2}\right]$.
    When combined with the previous results, we deduce that the condition holds for $\qt \in \left[2-\sqrt{2}, 1\right]$.
    With respect to \cref{eq:WaitingPtpZeroCostConstraintSecondSurroundGe}, again one can deduce from \cref{eq:EarlyDominatesFinalIndifferenceConditionSqrtIsPositive} that it is satisfied for $\qt \le \frac{1}{3.5}$, while the proof of \cref{prop:WaitingPtpZeroCostConstraintFirst} shows that the condition holds for $\qt \in \left[0, \frac{1}{2}\right]$.
    In total, observe that the intersection of the feasible parameters for \cref{eq:WaitingPtpZeroCostConstraintSecondSurroundLe,eq:WaitingPtpZeroCostConstraintSecondSurroundGe} is empty.

    \paragraphNoSkip{Feasible parameters for $\qo$}
    Finally, recall that the range \cref{eq:WaitingPtpZeroCostConstraintSecondRangeQo} includes feasible $\qo$ values when either one of \cref{eq:WaitingPtpZeroCostConstraintSecondLe1,eq:WaitingPtpZeroCostConstraintSecondGe0,eq:WaitingPtpZeroCostConstraintSecondSurround} is satisfied.
    Note that $0.58 \approx 2-\sqrt{2} > 0.5$, thus we get that \cref{eq:WaitingPtpZeroCostConstraintSecondRangeQo} is non-empty for $\qt \in \left[0,1\right]$.
    This implies that $\frac{\qo - 1}{1.5\qo + \qt - 2} \geq \frac{2\qo + \qt - 2}{2 \left(\qo - 1\right)}$ is satisfied when:
    \begin{align*}
        \qo
         &
        \in
        \left[
            \frac{1}{2}\left(
            3-3.5\qt
            -
            \sqrt{8.25\qt^2 -5\qt + 1}
            \right)
            ,
            \frac{1}{2}\left(
            3-3.5\qt
            +
            \sqrt{8.25\qt^2 -5\qt + 1}
            \right)
            \right]
        \\
        \qt
         &
        \in
        \left[0,1\right]
    \end{align*}
\end{proof}

\propWaitingPtpOne*
\begin{proof}
    From \cref{prop:MustardSet} and our assumption that $\ptp = 1$, we obtain two constraints: $\cost < 1+\frac{\qt}{2\left(\qo-1\right)}$ and $\qo < 1$.
    By substituting $\ptp = 1$ in \cref{eq:EarlyDominatesFinalIndifference}, we get $\left(\frac{\qo}{2} + \qt - 1\right)\cost + 1 - \qo < 0$, which we rearrange to: $1 - \qo < \left(1 - \frac{\qo}{2} - \qt\right)\cost$.
    Dividing the expression by $1 - \frac{\qo}{2} - \qt$ does not change the inequality's direction, as $\qo + \qt \in \left[0,1\right]$ and as in the current case we have $\qo > 0$.
    Thus, we deduce that $1 - \frac{\qo}{2} - \qt > 0$, and get a constraint on $\cost$: $\frac{1 - \qo}{1 - \frac{\qo}{2} - \qt} < \cost$.
    To satisfy both the new and the previous constraints, we require: $\frac{1 - \qo}{1 - \frac{\qo}{2} - \qt} < \cost < 1+\frac{\qt}{2\left(\qo-1\right)}$.
    Feasible $\cost$ values exist when:
    \begin{align*}
        0
         &
        <
        1+\frac{\qt}{2\left(\qo-1\right)}
        -
        \left(
        \frac{1 - \qo}{1 - \frac{\qo}{2} - \qt}
        \right)
        \nonumber \\&
        =
        \frac{2\qo + \qt - 2}{2\left(\qo-1\right)}
        -
        \left(
        \frac{1 - \qo}{1 - \frac{\qo}{2} - \qt}
        \right)
        \nonumber \\&
        =
        \frac{
            \left( 2\qo + \qt - 2 \right) \left( 1 - \frac{\qo}{2} - \qt \right)
            -
            2 \left( 1 - \qo \right)^2
        }{2 \left(\qo-1\right) \left( 1 - \frac{\qo}{2} - \qt \right)}
        \nonumber \\&
        =
        \frac{
            2 \left( 1 - \qo \right)^2
            -
            \left( 2\qo + \qt - 2 \right) \left( 1 - \frac{\qo}{2} - \qt \right)
        }{2 \left(1-\qo\right) \left( 1 - \frac{\qo}{2} - \qt \right)}
        \nonumber \\&
        =
        \frac{
            -\qo^2 + \left(2.5\qt + 1\right)\qo + \qt^2 - 3\qt
        }{2 \left(1-\qo\right) \left( 1 - \frac{\qo}{2} - \qt \right)}
    \end{align*}
    The denominator is positive, as $\qo \in \left(0,1\right)$, and as we have already shown that $1 - \frac{\qo}{2} - \qt > 0$.
    So, we only require:
    $0 < -\qo^2 + \left(2.5\qt + 1\right)\qo + \qt^2 - 3\qt$.
    Solving for $\qo$, the roots of the polynomial are:
    \begin{align*}
        \qo
         &
        =
        \frac{
            -\left(2.5\qt + 1\right)
            \pm
            \sqrt{
                \left(2.5\qt + 1\right)^2
                -
                4\cdot\left(-1\right)\cdot\left(\qt^2 - 3\qt\right)
            }
        }{-2}
        \nonumber \\&
        =
        \frac{
            -\left(2.5\qt + 1\right)
            \pm
            \sqrt{6.25\qt^2 + 5\qt + 1 + 4\qt^2 - 12\qt}
        }{-2}
        \nonumber \\&
        =
        \frac{5\qt + 2 \pm \sqrt{41\qt^2 - 28\qt + 4}}{4}
    \end{align*}
    The roots are in $\mathbb{R}$ when $0 \le 41\qt^2 - 28\qt + 4$, and the zeroes of this constraint are obtained at:
    \begin{equation}
        \label{eq:WaitingPtpOneSqrtZeroes}
        \qt
        =
        \frac{
            -(-28)
            \pm
            \sqrt{28^2 - 4\cdot41\cdot4}
        }{2\cdot41}
        =
        \frac{14 \pm 4\sqrt{2}}{41}
    \end{equation}
    For clarity, note that $\frac{14 - 4\sqrt{2}}{41} \approx 0.2$, $\frac{14 + 4\sqrt{2}}{41} \approx 0.48$.
    Thus, our parameter range is currently:
    $\qo \in \left(\frac{5\qt + 2 - \sqrt{41\qt^2 - 28\qt + 4}}{4}, \frac{5\qt + 2 + \sqrt{41\qt^2 - 28\qt + 4}}{4}\right)$, $\qt \in \left[0, \frac{14 - 4\sqrt{2}}{41}\right] \cup \left[\frac{14 + 4\sqrt{2}}{41},1\right]$, $\cost \in \left(\frac{1 - \qo}{1 - \frac{\qo}{2} - \qt}, 1+\frac{\qt}{2\left(\qo-1\right)}\right)$.
    For $\qo$ to be feasible, we require that the current feasible range has a non-empty intersection with $\left[0,1\right]$.
    This is analyzed in \cref{prop:WaitingPtpOneNonEmpty}, with the corresponding proof given in \cref{sec:Proofs}.
    \begin{restatable}[]{proposition}{propWaitingPtpOneNonEmpty}
        \label{prop:WaitingPtpOneNonEmpty}
        $\forall \qt \in \left[0, \frac{14 - 4\sqrt{2}}{41}\right)\cup\left(\frac{1}{2}, 1\right]: \left(\frac{5\qt + 2 - \sqrt{41\qt^2 - 28\qt + 4}}{4}, \frac{5\qt + 2 + \sqrt{41\qt^2 - 28\qt + 4}}{4}\right) \cap \left[0,1\right] \ne \emptyset$
    \end{restatable}
    To ensure that throughout the parameter space we furthermore require $\qo + \qt \in \left[0,1\right]$.
    An immediate calculation of the various possible extreme values of $\qt$ and the corresponding extrema of $\qo$ show that our feasible range for $\qt$ is: $\qt \in \left[0, \frac{14 - 4\sqrt{2}}{41}\right)$.
    This allows us to refine our feasible parameter range and conclude the proof:
    \begin{align*}
        \qo   & \in \left(\frac{5\qt + 2 - \sqrt{41\qt^2 - 28\qt + 4}}{4}, \frac{5\qt + 2 + \sqrt{41\qt^2 - 28\qt + 4}}{4}\right)
        \\
        \qt   & \in \left[0, \frac{14 - 4\sqrt{2}}{41}\right)
        \\
        \cost & \in \left(\frac{1 - \qo}{1 - \frac{\qo}{2} - \qt}, \frac{2\qo + \qt - 2}{2\left(\qo-1\right)}\right)
    \end{align*}
\end{proof}

\propWaitingPtpOneNonEmpty*
\begin{proof}
    For $\left(\frac{5\qt + 2 - \sqrt{41\qt^2 - 28\qt + 4}}{4}, \frac{5\qt + 2 + \sqrt{41\qt^2 - 28\qt + 4}}{4}\right) \cap \left[0,1\right] \ne \emptyset$ to hold, we require that both $\frac{5\qt + 2 - \sqrt{41\qt^2 - 28\qt + 4}}{4} < \frac{5\qt + 2 + \sqrt{41\qt^2 - 28\qt + 4}}{4}$, which is true for $\qt \in \left[0, \frac{14 - 4\sqrt{2}}{41}\right) \cup \left(\frac{14 + 4\sqrt{2}}{41}, 1 \right]$ due to \cref{eq:WaitingPtpOneSqrtZeroes}, and furthermore that at least one of \cref{eq:WaitingPtpOneFirst,eq:WaitingPtpOneSecond,eq:WaitingPtpOneThird} must hold.
        \begin{equation}
            \label{eq:WaitingPtpOneFirst}
            \frac{5\qt + 2 - \sqrt{41\qt^2 - 28\qt + 4}}{4} \in \left[0,1\right)
        \end{equation}
        \begin{equation}
            \label{eq:WaitingPtpOneSecond}
            \frac{5\qt + 2 + \sqrt{41\qt^2 - 28\qt + 4}}{4} \in \left(0,1\right]
        \end{equation}
        \begin{equation}
            \label{eq:WaitingPtpOneThird}
            \frac{5\qt + 2 - \sqrt{41\qt^2 - 28\qt + 4}}{4}
            \le 0 < 1 \le
            \frac{5\qt + 2 + \sqrt{41\qt^2 - 28\qt + 4}}{4}
        \end{equation}

        \paragraphNoSkip{Satisfying \cref{eq:WaitingPtpOneFirst}}
        Note that this condition holds when both \cref{eq:WaitingPtpOneFirstZero,eq:WaitingPtpOneFirstOne} hold.
        \begin{equation}
            \label{eq:WaitingPtpOneFirstZero}
            0 \le \frac{5\qt + 2 - \sqrt{41\qt^2 - 28\qt + 4}}{4}
        \end{equation}
        \begin{equation}
            \label{eq:WaitingPtpOneFirstOne}
            1 > \frac{5\qt + 2 - \sqrt{41\qt^2 - 28\qt + 4}}{4}
        \end{equation}
        The former is equivalent to $\sqrt{41\qt^2 - 28\qt + 4} \le 5\qt + 2$, and thus is satisfied for $\qt \in \left[0,3\right]$:
        \begin{align*}
            0
             &
            \geq
            41\qt^2 - 28\qt + 4 - \left(5\qt + 2\right)^2
            \nonumber \\&
            =
            41\qt^2 - 28\qt + 4 - \left(25\qt^2 + 20\qt + 4\right)
            \nonumber \\&
            =
            16\qt\left(\qt - 3\right)
        \end{align*}
        For \cref{eq:WaitingPtpOneFirstOne}, it is equivalent to $\sqrt{41\qt^2 - 28\qt + 4} > 5\qt - 2$, which is trivially satisfied when both $\qt < \frac{2}{5}$ and when the square root has a real solution, thus from \cref{eq:WaitingPtpOneSqrtZeroes} we get that $\qt \in \left[0, \frac{14 - 4\sqrt{2}}{41}\right)$, as $0.2 \approx \frac{14 - 4\sqrt{2}}{41} < \frac{2}{5}$.
    In case $\qt \geq \frac{2}{5}$, both sides of the constraint are positive, thus we can raise them to the power of two and solve to show that it is satisfied for $\qt > \frac{1}{2}$:
    \begin{align*}
        0
         &
        <
        41\qt^2 - 28\qt + 4 - \left(5\qt - 2\right)^2
        \nonumber \\&
        =
        41\qt^2 - 28\qt + 4 - \left(25\qt^2 - 20\qt + 4\right)
        \nonumber \\&
        =
        16\qt^2 - 8\qt
        \nonumber \\&
        =
        8\qt\left(2\qt - 1\right)
    \end{align*}
    In total, we get that the constraint holds when $\qt \in \left[0, \frac{14 - 4\sqrt{2}}{41}\right)\cup\left(\frac{1}{2}, 1\right]$.

            \paragraphNoSkip{Satisfying \cref{eq:WaitingPtpOneSecond}}
            This requirement is equivalent to satisfying both \cref{eq:WaitingPtpOneSecondZero,eq:WaitingPtpOneSecondOne}:
            \begin{equation}
                \label{eq:WaitingPtpOneSecondZero}
                0 < \frac{5\qt + 2 + \sqrt{41\qt^2 - 28\qt + 4}}{4}
            \end{equation}
            \begin{equation}
                \label{eq:WaitingPtpOneSecondOne}
                1 \geq \frac{5\qt + 2 + \sqrt{41\qt^2 - 28\qt + 4}}{4}
            \end{equation}
            \cref{eq:WaitingPtpOneSecondZero} is satisfied when $\qt \in \left[0, \frac{14 - 4\sqrt{2}}{41}\right] \cup \left[\frac{14 + 4\sqrt{2}}{41},1\right]$, due to \cref{eq:WaitingPtpOneSqrtZeroes} and as $\qt$ is a probability and thus must be in $\left[0,1\right]$.
            On the other hand, \cref{eq:WaitingPtpOneSecondOne} is equivalent to $\sqrt{41\qt^2 - 28\qt + 4} \le 2 - 5\qt$, which does not hold for $\qt > 0.2$.
            For $\qt \le 0.2$, we raise the inequality to the power of two and deduce that the condition is satisfied when $\qt \in \left[0, 0.2\right]$:
            \begin{align*}
                0
                 &
                \geq
                41\qt^2 - 28\qt + 4 - \left(2 - 5\qt\right)^2
                \nonumber \\&
                =
                41\qt^2 - 28\qt + 4 - \left(4 - 20\qt + 25\qt^2\right)
                \nonumber \\&
                =
                16\qt\left(\qt - \frac{1}{2}\right)
            \end{align*}
            In total, \cref{eq:WaitingPtpOneSecond} holds when $\qt \in \left[0, 0.2\right]$.

            \paragraphNoSkip{Satisfying \cref{eq:WaitingPtpOneThird}}
            The current condition is equivalent to both of the following.
            \begin{equation}
                \label{eq:WaitingPtpOneThirdLe}
                \frac{5\qt + 2 - \sqrt{41\qt^2 - 28\qt + 4}}{4} \le 0
            \end{equation}
            \begin{equation}
                \label{eq:WaitingPtpOneThirdGe}
                \frac{5\qt + 2 + \sqrt{41\qt^2 - 28\qt + 4}}{4} \ge 1
            \end{equation}
            \cref{eq:WaitingPtpOneThirdLe} complements \cref{eq:WaitingPtpOneFirstZero}, implying that the requirement does not hold for $\qt \in \left[0,1\right]$, and that furthermore \cref{eq:WaitingPtpOneThird} has no feasible satisfying parameters, as its satisfiability hinges on both \cref{eq:WaitingPtpOneThirdLe,eq:WaitingPtpOneThirdGe} holding at the same time.

            \paragraphNoSkip{Feasible parameters for $\qt$}
            As we require that at least one out of \cref{eq:WaitingPtpOneFirst,eq:WaitingPtpOneSecond,eq:WaitingPtpOneThird} holds, we take the union of the feasible parameter range of each of these, which equals:
            $
            \qt \in \left[0, \frac{14 - 4\sqrt{2}}{41}\right)\cup\left(\frac{1}{2}, 1\right]
    $.
\end{proof}

\section{Additional Related Work}
\label[appendix]{sec:AdditionalRelatedWork}
In this section, we go over additional related papers.
The reliance of on-chain governance protocols on tokens to represent voting power bears a similarity to corporate governance mechanisms which give voting power per share.
Thus, we also go over works on corporate governance which could interest the blockchain community.

\subsection{Blockchain Governance}
\paragraphNoSkip{Theoretical Studies}
Kiayias and Lazos~\cite{kiayias2022sok} present a \gls{SoK} of blockchain governance mechanisms.
In the paper, governance is defined in the broad sense, taken to mean both the mechanism by which improvements are suggested and adopted by cryptocurrencies like Bitcoin, and the on-chain protocols employed by \gls{DeFi} projects.
In the paper, several properties of interest are defined, such as who is granted suffrage by each mechanism, and the confidentiality of the voting process, and quantitatively evaluate the various mechanisms.

Wright~\cite{wright2019quadratic} suggest using Posner and Weyl's~\cite{posner2014quadratic} quadratic voting in blockchains.
In such voting mechanisms, the voting power of an individual is equal to the square root of the number of votes it casts.
Dimitri~\cite{dimitri2022quadratic} analyze a \gls{PoS} consensus mechanism which relies on quadratic voting, and note it is not Sybil resistant: voters can gain a disproportionate advantage by splitting their votes across multiple identities.
Indeed, similar mechanisms are in use in practice~\cite{team2022axelar}.
Some practitioners~\cite{rodriguez2022quadratic} recommend improving Sybil-resistance by randomly choosing the subset of eligible voters for each proposal.
We note that adversaries can circumvent this measure by splitting their funds into multiple accounts in advance.

\paragraphNoSkip{Empirical Studies}
Fritsch~\emph{et al.}~\cite{fritsch2022analyzing} study the distribution of voting power and the degree of decentralization of the governance mechanisms used by Compound, Uniswap and ENS.
This was done by evaluating various measures, such as the Gini coefficient~\cite{yitzhaki2013gini} of token holders, and the frequency with which single voters have influenced decisions.
The authors found that the surveyed protocols became more decentralized with time, and that a small number of voters hold a considerable amount of voting power, but that they do not exercise it often.

Feichtinger~\emph{et al.}~\cite{feichtinger2023hidden} later extended the study of \cite{fritsch2022analyzing} using the same metrics, but with more recent data spanning a larger amount of governance protocols.
They show that governance mechanisms are in fact still relatively centralized, with majority power concentrated in the hands of few users in $17$ mechanisms out of the total of $21$ which were examined.
These results are supplemented by an analysis of the costs associated with governance protocols, which are estimated at over $\$10$ Million.
    They consider both costs directly related to participation in decision-making (e.g., paying fees for voting transactions), and indirect ones (costs paid due to the overhead of updating governance-related data structures when performing unrelated actions, like token transfers).

    Fan~\emph{et al.}~\cite{fan2022towards} use the case of Uniswap and SushiSwap as a natural experiment to measure how governance token issuance incentivizes the migration of users between two competing platforms.
    In particular, SushiSwap was created as a fork of Uniswap, and thus by design both are identical, with the exception that SushiSwap rewarded liquidity providers with governance tokens, while Uniswap did not.
    The authors note that although SushiSwap's policy quickly attracted liquidity providers, some of these providers tended to remove their liquidity a few days afterwards, especially ones who provide large amounts of liquidity.

    Barbereau~\emph{et al.}~\cite{barbereau2022defi} quantify the centralization of the governance mechanisms of Uniswap, Maker, SushiSwap, Yearn Finance and UMA by giving a high-level comparison of the token issuance strategy of each, and then performing an empirical evaluation of the distribution of tokens among users, concluding that voting rights are centralized in all examined projects.
    With respect to issuance strategies, the authors found that only SushiSwap and Yearn Finance adopted ``fair'' strategies that do not involve some preliminary issuance of tokens by the projects' creators and then selling it to users.
    They note that a fair strategy could be to the detriment of a project, and give the case of Yearn Finance as an example: the project's lack of centrally owned capital made it difficult to fund a third-party security audit.
    The paper proceeds with an empirical evaluation of the distribution of tokens among users through the lens of various metrics, such as the Gini coefficient, and the distributions' similarity to a ``good'' distribution that represents an ideal level of decentralization~\cite{gochhayat2020measuring}.

    Barbereau~\emph{et al.}~\cite{barbereau2022decentralised} expand the study of \cite{barbereau2022defi} by analyzing user participation in the surveyed mechanisms, and including additional mechanisms in their analysis, specifically 0x, Aave, Compound and Uniswap.
    They found that typically, the rate of participation is low, only rarely rising above a turnout of $1\%$.
    The authors ascribe the low rate to the ability of users to trade them, and justify this conclusion by showing that until January 2022, the trading volume of governance tokens is much higher when compared to the amount of tokens which are delegated.

    Sun~\emph{et al.}~\cite{sun2022decentralization} attempt to quantify the level of centralization in Maker's governance protocol, using data spanning August 2019 and October 2021, and by applying metrics like voter participation and Twitter sentiment.
    The authors find that the average and median share of votes cast by the largest voter per proposal were $52.66\%$ and $48.35\%$, respectively.
    Furthermore, they examine the financial impact of centralization.
    They conduct regression analyses and show that Maker's trading volume and the number of transactions interacting with the platform are positively correlated with the share of users in possession of $10-100$K MKR tokens, and negatively correlated with the share of those who hold more than $100$K tokens.

    Sharma~\emph{et al.}~\cite{sharma2023unpacking} conducted a mixed-methods study to evaluate the effectiveness of governance mechanisms employed by $10$ \glspl{DAO}: Compound, BitDAO, AssangeDAO, BanklessDAO, KrauseHouse, LivePeer, MetaGammaDelta, MolochDAO, dxDAO.
    For their work, they interviewed $10$ participants of such mechanisms, and used responses to identify decentralization metrics which are considered as relevant by the community, primarily, the distribution of token ownership, token holders' participation rate in the governance process, and the geographic distribution of holders.
    By these metrics, the authors find that Compound, AssangeDAO, Bankless and Krausehouse fare poorly with respect to decentralization, and show that across all examined platforms, proposal success is positively correlated with the token holdings of the proposer, and that repeated participation in votes is positively correlated with a voter's holdings.
    Morevoer, the work examines the technical aspects of governance mechanisms, and shows that some are limited in that they do not allow successful proposals to execute arbitrary transactions, while others rely on third-party services to fulfil proposals.

    Dotan~\emph{et al.}~\cite{dotan2023vulnerable} perform a series of case-studies on vulnerabitilies of governance mechanisms and how these were exploited.
    For example, the authors show that Compound's mechanism awarded users with governance tokens for taking loans, where the reward mechanism resulted in the interest-rate of the corresponding loans to be negative.
    The case studies are augmented by a data-driven analysis showing that voter turnout is negatively correlated with transaction fees, and that on some occasions, voters possessed governance tokens for the span of a single voting window.

    Kitzler~\emph{et al.}~\cite{kitzler2023governance} study so-called \gls{DAO} ``contributors'', such as \gls{DAO} owners and developers.
    They find that contributors posses the majority of voting power in $7.5\%$ of the examined governance mechanisms, and have decided at least a single proposal in $20.4\%$ of them.
    The authors uncover systemic evidence in support of Dotan~\emph{et al.}~\cite{dotan2023vulnerable}, showing that in $14.8\%$ of the examined governance proposals, the volume of trades of governance tokens made in the $100$ days preceding each proposal exceeded the amount of votes required to obtain a majority in the corresponding proposal.

    \subsection{Corporate Governance}
    \paragraphNoSkip{Theoretical Studies}
    Cvijanovic~\emph{et al.}~\cite{cvijanovic2020free} propose a theoretical model that captures the ``free-rider effect'', wherein voters who have popular preferences are less inclined to vote.
    The work shows that there is an equilibrium where all ``underdog'' voters vote and other voters only participate partially.

    Dekel and Wolinsky~\cite{dekel2011buying} study the efficiency of the outcome of corporate control contests when allowing voting rights to be traded separately of shares, where an outcome was considered efficient if the actor that wins control of the corporation is the one that maximizes shareholder value.
    The authors show that allowing vote trading always results in inefficient outcomes, and note that this contrasts with the intuition that in the political setting vote buying should be considered deterimental, while in the corporate setting it should enhance efficiency.

    A method for assessing the value of shareholder voting rights was given by Kalay~\emph{et al.}~\cite{kalay2014market}.
    Their method compares the price of a share to a ``synthesized'' non-voting share, with the latter being constructed using put-call parity.
    The method is then employed on real-world data, showing that the value of voting rights increases around special shareholder meetings and when voting on contentious mergers and acquisitions.

    \paragraphNoSkip{Empirical Studies}
    Brav~\emph{et al.}~\cite{brav2022retail} conduct an study on retail shareholder voting, finding that turnout is higher for voters with larger holdings or when the expected benefits from voting are high, while turnout decreases when the cost of participation is high.
    An analysis of shareholder voting in $17$ European countries between the years $2005$ and $2010$ was carried out by Renneboog and Szilagyi~\cite{renneboog2013shareholder}.
    The authors found that turnout by owners of free-float shares (that is, shares that can be publicly traded) was relatively low: $17\%$ in France, $10\%$ in Germany, and $4\%$ in Italy.
Furthermore, the authors note that shareholder proposals (proposals written by shareholders rather than by the management of a corporation) were mostly raised by shareholders of large and poorly performing corporations.
Additionally, dissenting votes were more frequent when the proposals advanced the adoption of anti-takeover devices and executive compensation.
A higher rate of dissent in proposals on remuneration was also found by Hewitt~\cite{hewitt2011exercise}, who studied OECD countries.

Li and Schwartz-Ziv~\cite{li2018how} show that the daily volume, number of trades and the volatility of shares are increased around shareholder voting dates, particularly when the votes are considered important and when voting culminated in a dissenting outcome.
A similar finding that share prices decline after voting is given by Fos and Holderness~\cite{fos2022distribution}.
Additionally, the work presents data showing that corporations sometimes reveal voting dates after the fact, thereby preventing stockholders from exercising their right to vote.
We note that commonly used governance mechanisms do not allow such paractices, as proposals are announced publicly on the blockchain.

\paragraphNoSkip{Shareholder Value Maximization (SVM)}
\emph{\Gls{SVM}} is a business goal which can be broadly defined as striving to increase the wealth of the business' shareholders.
Sundaram and Inkpen~\cite{sundaram2004corporate} propose \gls{SVM} as the ``preferred'' goal for corporate finance, arguing it is ethical and easily measurable.
Freeman~\emph{et al.}~\cite{freeman2004stakeholder} criticize \gls{SVM} as it only focuses on shareholders and ignores corporate stakeholders, with the latter taken to mean any actor who can affect or be affected by the corporation.
Thus, shareholders are stakeholders, but not necessarily the other way around.
For example, employees of the corporation are stakeholders, even if they do not possess any shares.
Inkpen and Sundaram~\cite{inkpen2022endurance} discuss the relevance of \gls{SVM} in the face of academic criticism and later developments in the corporate finance world, such as the supposed greater divergence in shareholder interests in recent years~\cite{goranova2022corporate}.

\section{Glossary}
\label[appendix]{sec:Glossary}
A summary of all symbols and acronyms used in the paper.
\setglossarystyle{alttree} \glssetwidest{AAAA} 
\printnoidxglossary[type=symbols]
\printnoidxglossary[type=\acronymtype]

\end{document}